\documentclass{fac}

\journal{ArXiV version}

\usepackage{lineno,hyperref}
\modulolinenumbers[5]

\usepackage{times}
\usepackage{amsfonts,amstext,amsmath,amssymb}
\usepackage{stmaryrd}
\usepackage{xcolor}
\usepackage{epsfig}
\usepackage{picinpar}
\usepackage{bm}

\usepackage{oz}\renewcommand{\spot}{\mathrel{\cdot}}
\usepackage{traces}
\newcommand{\Nil}{TMPNIL}
\usepackage{rgref}
\usepackage{ifthen}

\newcommand{\draftonly}[1]{}

\newcommand{\anegate}{\cnegate}
\newcommand{\tnegate}{\neg}

\newcommand{\atomid}{\boldsymbol{\epsilon}}

\newcommand{\AtomicSteps}{\mathcal{A}}
\newcommand{\Commands}{\mathcal{C}}
\newcommand{\Tests}{\mathcal{T}}
\newcommand{\BooleanAlg}{\mathcal{B}}
\newcommand{\sync}{\mathbin{\otimes}}

\newcommand{\Identity}[1]{\mathbf{1_{#1}}}
\newcommand{\syncid}{\Identity{}}
\newcommand{\Identitycommand}[1]{\mathit{\mathcal{I}d_{#1}}}
\newcommand{\syncidcommand}{\Identitycommand{}}

\newcommand{\join}{\mathop{\sqcup}}
\mathchardef\mhyphen="2D
\newcommand{\pguard}[1]{(\mathop{\Keyword{\pi \mhyphen restrict}} #1)}
\newcommand{\eassume}[1]{(\mathop{\Keyword{\epsilon \mhyphen assm}} #1)}

\newcommand{\pibot}{\boldsymbol{\pi}} 
\newcommand{\ebot}{\boldsymbol{\epsilon}}

\newcommand{\e}{\varepsilon}

\newcommand{\prednegate}[1]{\overline{#1}}

\newcommand{\PCommands}{{\cal P}}
\newcommand{\ECommands}{{\cal E}}

\makeatletter
\def\RSpec{\@ifnextchar*{\@RSpec}{\@@RSpec}}
\def\@RSpec*#1#2#3{\ifx\@empty#1\else#1\colon\fi
   [#2\ifx\@empty#2\else,~\fi#3]}
\def\@@RSpec#1#2#3{\ifx\@empty#1\else
   \begin{array}{@{}l@{}}#1\colon\end{array}\!\!\fi%
   \left[\begin{array}{@{}c@{}}#2\end{array}\ifx\@empty#2\else,~\fi
   \begin{array}{@{}c@{}}#3\end{array}\right]}
   
\def\Spec{\@ifnextchar*{\@Spec}{\@@Spec}}
\def\@Spec*#1#2#3{\ifx\@empty#1\else#1\colon\fi
   \llparenthesis#2\ifx\@empty#2\else,~\fi#3\rrparenthesis}
\def\@@Spec#1#2#3{\ifx\@empty#1\else
   \begin{array}{@{}l@{}}#1\colon\end{array}\!\!\fi%
   \llparenthesis\begin{array}{@{}c@{}}#2\end{array}\ifx\@empty#2\else,~\fi
   \begin{array}{@{}c@{}}#3\end{array}\rrparenthesis}
\makeatother

\newcommand{\quint}[5]{\{#1,#2\}~#5~\{#3,#4\}}
\newcommand{\quintprgqc}{\quint{p}{r}{g}{q}{c}}

\renewenvironment{proof}{\par\noindent\textbf{Proof.}}{\hfill$\Box$}
\newcommand{\Pevent}[1]{\cpstep{#1}}

\renewcommand{\interleave}{\mathop{|\hspace*{-0.2ex}|\hspace{-0.2ex}|}}

\newcommand{\refeqn}[1]{(\ref{eqn-#1})}
\newcommand{\refsect}[1]{Sect.~\ref{S-#1}}
\newcommand{\reflemma}[1]{Lemma~\ref{L-#1}}

\newcommand{\synchpstep}[1]{\Pevent{#1}}

\newcommand{\synchpstepE}{\synchpstepd}
\newcommand{\synchestepE}{\atomid}

\newcommand{\spstep}[1]{\pstep{#1}}
\newcommand{\sestep}[1]{\estep{#1}}

\newcommand{\sestepE}{\atomid}

\newcommand{\ssilent}{\spstep{\silent}}

\newcommand{\pl}{\parallel}

\newcommand{\plE}[1]{\underset{#1}{\pl}}

\newcommand{\Event}{Event}

\newcommand{\silent}{\iota}

\newcommand{\compl}[1]{\overline{#1}}

\newcommand{\ecompl}[1]{\tilde{#1}}

\newcommand{\atev}[1]{\atomic{#1}}

\newcommand{\atevs}{\atev{\silent}}

\newcommand{\Hide}[2]{{#2}/{#1}}
\newcommand{\Res}[2]{{#2\backslash#1}}

\newcommand{\commentout}[1]{}

\definecolor{Red}{rgb}{1.,0.,0.}
\definecolor{Blue}{rgb}{0.,0.,1.}
\definecolor{Pink}{rgb}{1.,0.75,0.8}
\definecolor{Green}{rgb}{0.2.,0.5,0.2}

\newtheorem{theorem}{Theorem}
\newtheorem{corollary}[theorem]{Corollary}
\newtheorem{lemma}[theorem]{Lemma}

\newcommand{\figurerule}{\rule{\textwidth}{0.5pt}}

\edef\today{\number\day\ \ifcase\month\or
  January\or February\or March\or April\or May\or June\or
  July\or August\or September\or October\or November\or December\fi
  \ \number\year}
\newcounter{Hours}
\newcounter{Minutes}
\newcommand{\CurrentTime}{%
 \ifthenelse{\value{Hours}<10}{0}{}\theHours:%
 \ifthenelse{\value{Minutes}<10}{0}{}\theMinutes}
\newcommand{\runningdate}{\draftonly{(\today\ \CurrentTime\ DRAFT)}}

\begin{document}

\draftonly{
\let\origthepage=\thepage
\makeatletter
\renewcommand{\thepage}{\@arabic\c@page-R}
\makeatother
\input{conf-diffs}
\let\thepage=\origthepage
\clearpage
\setcounter{page}{1}
\setcounter{section}{0}
}

\title[Synchronous refinement algebra \runningdate]{A synchronous program algebra:\\
a basis for reasoning about shared-memory and event-based concurrency}
\author[I. J. Hayes et al.  \runningdate]{
Ian J. Hayes \and 
Larissa A. Meinicke \and 
Kirsten Winter \and 
Robert J. Colvin 
\\
School of Information Technology and Electrical Engineering, \\
The University of Queensland, Australia}
\correspond{Ian J. Hayes, School of Information Technology and Electrical Engineering,\\ 
The University of Queensland, Australia 4072.
              e-mail: Ian.Hayes@itee.uq.edu.au \\
This work was supported by Australian Research Council (ARC) Discovery Project DP130102901.
}
\date{\today\ \CurrentTime}

\maketitle
\makecorrespond

\begin{abstract}
This research started with 
an algebra for reasoning about rely/guarantee concurrency for a shared
memory model. The approach taken led to a more \emph{abstract algebra of
atomic steps}, in which atomic steps synchronise (rather than
interleave) when composed in parallel. The algebra of rely/guarantee
concurrency then becomes an instantiation of the more abstract
algebra. Many of the core properties needed for rely/guarantee
reasoning can be shown to hold in the abstract algebra where their
proofs are simpler and hence allow a higher degree of automation. 
The algebra has been encoded in Isabelle/HOL to provide a basis
for tool support for program verification.

In rely/guarantee concurrency, programs are specified to guarantee certain behaviours 
until assumptions about the behaviour of their environment are violated. 
When assumptions are violated, program behaviour is unconstrained (aborting), 
and guarantees need no longer hold.
To support these guarantees a second synchronous operator, weak conjunction, was introduced:
both processes in a weak conjunction must agree to take each atomic step, 
unless one aborts in which case the whole aborts.
In developing the laws for parallel and weak conjunction 
we found many properties were shared by the operators and
that the proofs of many laws were essentially the same.
This insight led to the idea of generalising synchronisation 
to an abstract operator with only the axioms that are shared 
by the parallel and weak conjunction operator, 
so that those two operators can be viewed as instantiations of the abstract synchronisation operator.
The main differences between parallel and weak conjunction are how they combine
individual atomic steps; 
that is left open in the axioms for the abstract operator.

Milner's process algebra SCCS 
also includes a synchronous parallel operator
and (again) the main difference between it and Aczel's synchronous parallel operator 
is how it combines individual atomic steps.
Milner's parallel can be seen as another instance of the abstract synchronisation operator.
Moreover, the realisation that the synchronisation
mechanisms of standard process algebras, such as CSP and CCS/SCCS, can
be interpreted in our abstract algebra gives evidence of its unifying
power. 
\end{abstract}

\section{Introduction}

Our overall goal is to provide mechanised support for the development and verification
of concurrent programs 
based on the rely/guarantee technique of Jones~\cite{Jones81d,Jones83a,Jones83b}.
Our approach is to develop a concurrent program algebra
similar to Concurrent Kleene Algebra \cite{DBLP:journals/jlp/HoareMSW11}
and Synchronous Kleene Algebra \cite{Pris10}.
The algebra is mechanised as Isabelle/HOL theories.
The theory makes use of a synchronous parallel operator due to Aczel~\cite{Aczel83,DeRoever01}
that is designed to support the rely/guarantee approach for shared memory concurrency.
The general nature of the algebra lends itself, unexpectedly to the authors
at first, to an interpretation in which properties of the abstract communication of process algebras 
may be described, and which uses concepts common to the rely/guarantee
framework.

\paragraph{Atomic steps and commands}
Many models of concurrent systems are based on a notion of primitive \emph{atomic} steps
which can be combined into more complex terms, referred to here as \emph{commands}.
For process algebras like Milner's CCS/SCCS \cite{CaC} and Hoare's CSP \cite{Hoare85}, 
events are viewed as being atomic. 
For shared memory concurrency, 
operations like individual reads and writes of memory words are viewed as atomic.

\paragraph{Synchronisation}
In CSP, events common to the alphabets of two parallel processes are synchronised as a single event 
of their composition, while other events of the processes interleave \cite{Hoare85}. 
In CCS, every event $e$ has a complementary event $\ecompl{e}$, where a parallel composition may 
synchronise an $e$ event from one process with an $\ecompl{e}$ event of the other to give an
internal action $\iota$, or the actions may interleave \cite{CaC}. 
Milner recognised the commonality between these languages when he proposed Synchronous 
CCS (SCCS) \cite{Milner83}.  In SCCS every step of a parallel composition $c \parallel d$  
is a synchronisation of a step of $c$ with a step of $d$. To allow for interleaving, 
a step $e$ of one process may be matched with  an idling step, $\mathbf{1}$, of the other process 
to give a step $e$ of their parallel composition. 
CCS can be encoded within SCCS by providing an asynchronising operation that allows
interleaving of atomic steps \cite[Section 9.3]{CaC}.

\paragraph{Rely/guarantee concurrency}
To handle the semantics of rely/guarantee concurrency for shared memory systems 
Aczel \cite{Aczel83,DeRoever01} invented traces that include both program steps 
$\pstepd(\sigma,\sigma')$ and environment steps $\estepd(\sigma,\sigma')$,
where the state space contains the values of the program variables and 
$\sigma$ and $\sigma'$ are the before and after states of the steps.
The environment steps of a process record the interference from its environment. 
In Aczel's model a parallel composition $c \parallel d$ synchronises a program step 
$\pstepd(\sigma, \sigma')$ of $c$ with an environment step $\estepd(\sigma,\sigma')$ of $d$ 
to give a program step $\pstepd(\sigma, \sigma')$ of $c \parallel d$. 
It also synchronises a common environment step $\estepd(\sigma,\sigma')$ of both $c$ and $d$ to 
give an environment step $\estepd(\sigma,\sigma')$ of $c \parallel d$. 
Hence Aczel's shared memory model, like Milner's SCCS, 
also treats parallel composition as a synchronising operator. 

\paragraph{Abstracting synchronous operators}
In our previous work on an algebra of synchronous atomic steps \cite{FM2016atomicSteps},
our motivation was to provide an algebra that supported rely/guarantee concurrency and
Aczel's synchronous model of the parallel operator.
It was only after we developed the algebra that we realised that it was related to
Milner's SCCS and hence also to Prisacariu's Synchronous Kleene Algebra (SKA) \cite{Pris10}.
Our previous research on rely/guarantee concurrency \cite{HayesJonesColvin14TR,FACJexSEFM-14,AFfGRGRACP} 
also made use of a weak conjunction operator $\together$
(explained in more detail in Section \ref{S-instantiations})
that is also a synchronous operator with properties similar to the  parallel operator.
One of the contributions of this paper is to devise an abstract algebra for such synchronous operators.

Perhaps the most important concept for abstracting synchronous operators is that the 
algebra for synchronising the individual atomic steps 
can be treated separately to the algebra for the wider context that supports sequences of atomic steps.
For Milner's SCCS, the (atomic) events form a commutative group with an operator $\times$ 
\cite[Section 9.3]{CaC}.
He used a prefixing operator $a.c$ to prefix an event $a$ onto a process $c$,
and his parallel operator satisfies the interchange law,
\begin{equation}\label{Milner-interchange}
  a.c \parallel b.d = (a \times b).(c \parallel d)~,
\end{equation}
where $a$ and $b$ are events and $c$ and $d$ are processes.
Our approach differs in that we embed atomic commands 
(corresponding to Milner's events)
as a distinguished subset of commands
and thus make use of the same parallel operator to combine them
(rather than introducing a new operator $\times$).
In addition, prefixing becomes a special case of sequential composition
for which the first command is atomic. 
Milner's interchange law (\ref{Milner-interchange}) becomes,
\begin{equation}\label{parallel-interchange-sequential}
  a \Seq c \parallel b \Seq d = (a \parallel b) \Seq (c \parallel d)~,
\end{equation}
where it should be emphasised that this law only provides an equality if $a$ and $b$ are atomic commands.

One of the advantages of treating atomic steps as a subset of commands is that
all of the operators on commands may be applied to atomic commands.
In particular, because commands form a lattice under the refinement ordering,
the lattice meet and join operators can be applied to atomic commands.
In fact, atomic commands form a sub-lattice of commands.
It also turns out to be useful to add a complement operator on atomic commands,
which means they form a Boolean algebra (see Section \ref{S-atomic}).

The main differences between the different forms of synchronous operators 
(e.g. the different parallel operators in the different languages and weak conjunction)
is how they are defined on atomic steps
but otherwise these operators satisfy a range of similar laws 
and hence, to avoid duplicating these laws and their proofs,
it is advantageous to consider an algebra for an abstract synchronisation operator $\sync$,
which we latter instantiate for each synchronous operator.
The operator $\sync$ is associative, commutative and has identity $\syncidcommand$.

The abstract synchronisation operator $c\sync d$ synchronises the commands $c$ and $d$ by 
synchronising the atomic steps of $c$ and $d$. 
At this level of abstraction the definition of $a\sync b$ for atomic steps $a$ and $b$ is left open 
because it differs for the different instantiations of the synchronisation operator (e.g., 
parallel composition and weak conjunction). 
In fact, the definition of $a \sync b$ on atomic commands for each instantiation 
largely defines the respective operation. 
Given that the command $\Nil$ is the null command, i.e.~the identity of sequential composition,
$\sync$ satisfies,
\begin{eqnarray}
  a \Seq c \sync b \Seq d & = & (a \sync b) \Seq (c \sync d) \label{sync-interchange-sequential} \\
  \Nil \sync \Nil & = & \Nil  \label{nil-sync-nil} \\
  a \Seq c \sync \Nil & = & \top \label{atomic-sync-nil}
\end{eqnarray}
where $\top$ is the top of the refinement lattice and represents the everywhere infeasible command
(sometimes referred to as ``magic'').
For example, the above laws hold with $\sync$ instantiated with $\parallel$,
so that (\ref{sync-interchange-sequential}) corresponds to (\ref{parallel-interchange-sequential}).
The other laws codify that two null commands synchronise to give null (\ref{nil-sync-nil})
and that a null command cannot synchronise with a command that performs at least one atomic step 
-- their combination is infeasible (\ref{atomic-sync-nil}).

\paragraph{Distinguishing program and environment steps}
Motivated by Aczel's model \cite{Aczel83},
when dealing with the parallel operator it is useful to consider 
two disjoint subsets of atomic commands:
one corresponding to program steps
and 
the other corresponding to environment steps,
where there is usually exactly one environment step to match each and every program step.
These two sets form a basis for building any atomic command as a non-deterministic choice
over a set of program and environment steps.
The set of all commands built from program steps forms a sub-lattice,
as does the commands build from environment steps.

It turns out that the distinction between program and environment steps is also
relevant for handling process algebras.
For a set of events $Event$ and $e \in Event$, let $\spstep{e}$ represent the process performing event $e$
and $\sestep{e}$ represent the environment of the process performing an $e$.
For CSP-style concurrency,
two processes synchronising on an event $e$ corresponds to a $\spstep{e}$ step of one process 
synchronising with a $\spstep{e}$ step of the other.
And an interleaving event corresponds to a $\spstep{e}$ step of one process synchronising 
with a $\sestep{e}$ step of the other.

The non-deterministic choice over all environment steps, 
$\ebot = \Nondet_{e \in Event} \epsilon(e)$,
allows the environment to perform any event whatsoever
and hence corresponds to the atomic step identity of parallel
and corresponds to Milner's identity $\textbf{1}$ for SCCS \cite{CaC}.
However, Milner only supports the identity and not the complete range of commands
that can be built from combinations of program and environment steps.

For a set of events, $E \subseteq Event$, it is useful to define 
an atomic command that can perform any of the events in $E$, 
$\cpstep{E} = \Nondet_{e \in E} \spstep{e}$,
and 
another that allows the environment to perform any events in $E$,
$\cestep{E} = \Nondet_{e \in E} \sestep{e}$.
These commands can be used as primitives in representations of process algebras
(Section~\ref{S-process-algebras}).
Returning to Aczel's model, the primitive events are pairs of before and after states, $(\sigma,\sigma')$,
so that $\pi(\sigma,\sigma')$ forms a primitive atomic program step and
$\epsilon(\sigma,\sigma')$ forms a primitive atomic environment step.
In this case a set of events becomes a set of pairs of states,
i.e.\ a binary relation $r$ on states.
This gives the commands $\cpstep{r}$ and $\cestep{r}$ that we use in Section~\ref{S-rely-guarantee}
to build our rely/guarantee theory.
This simple approach of treating a pair of states as an event
shows part of the relationship between process algebras and shared memory concurrency.
There is a further constraint in shared memory concurrency
that the after state of one step equals the before state of the next.

\paragraph{Overview of paper}
Section~\ref{S-general-algebra} introduces a Demonic Refinement Algebra (DRA)
that includes a sub-algebra of tests
(Section~\ref{S-tests}),
and a sub-algebra of atomic steps (Section~\ref{S-atomic}).
Section~\ref{S-sync-algebra} extends this with our abstract synchronisation operator ($\sync$)
along with a range of laws applicable at this level of abstraction,
in particular, those involving iterations of atomic steps (see Section~\ref{S-atomic-iteration}).
Section~\ref{S-instantiations} then develops two instantiations of the abstract synchronisation operator 
as a synchronous parallel operator and as weak conjunction.
Section~\ref{S-pi-env} examines introducing sub-algebras of program and environment steps
and using them as a basis for all atomic commands.
Section~\ref{S-rely-guarantee} develops abstract specification, rely and guarantee commands that can be used to specify rely/guarantee quintuples.
Section~\ref{S-rg-logic} provides an interpretation of these abstract commands to support
rely-guarantee reasoning about shared memory concurrency.
It contains a formal development of an abstract parallel introduction law,
a fundamental law in rely-guarantee concurrency.
Section~\ref{S-process-algebras} provides an interpretation of the algebra that provides process algebraic abstract communication via
events.  We show how the binary synchronisation of CCS and the multi-way synchronisation of CSP are defined in the abstract algebra.  We
also sketch a simple interpretation of the algebra that combines both state-based and event-based communication.
Section~\ref{S-related-work} considers related work.

\begin{figure}
 \centering
   \includegraphics[viewport=31 606 332 828,width=0.48\linewidth]{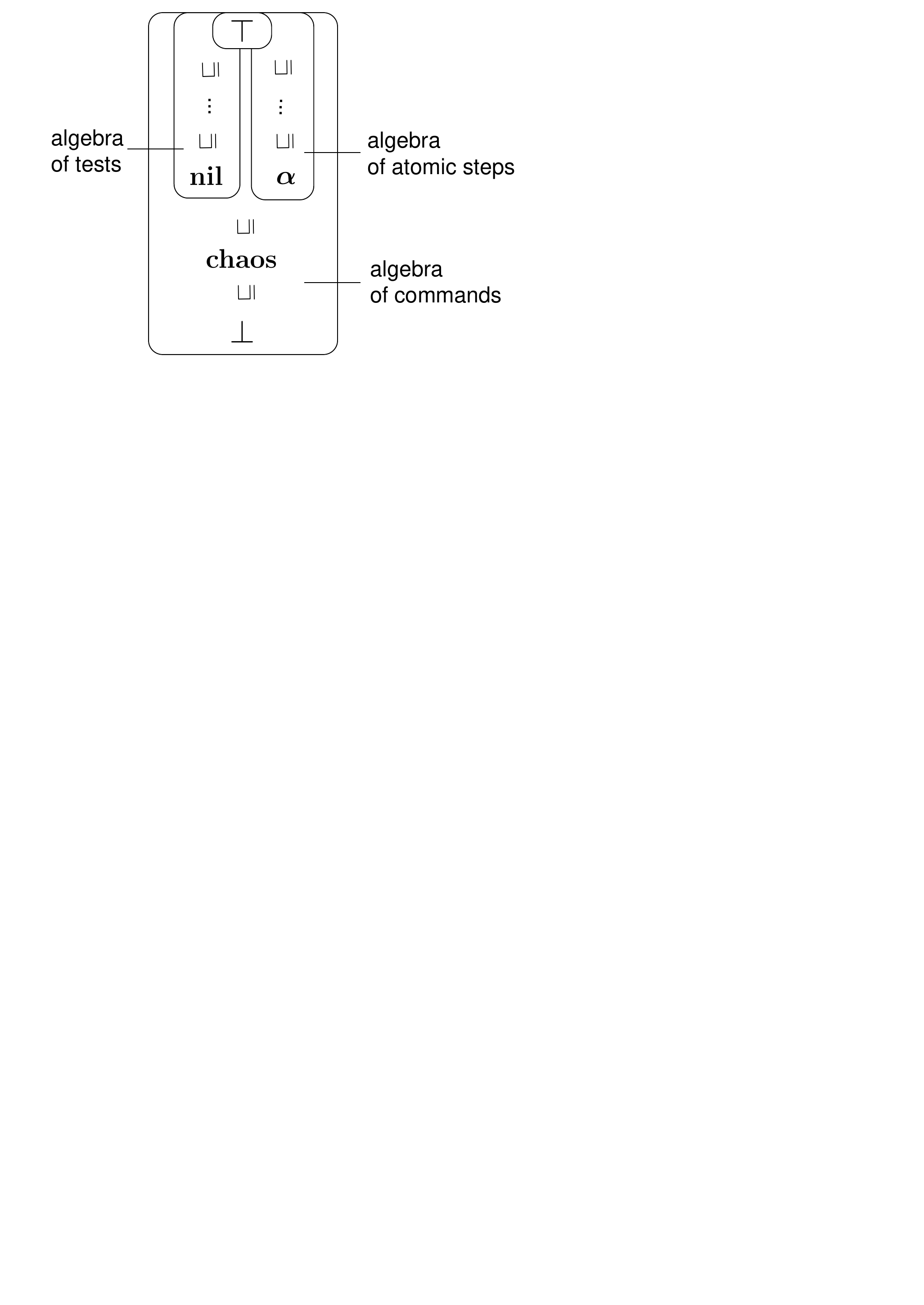}
  \caption{The Synchronous Refinement Algebra and its sub-algebras}
  \label{lattices}
\end{figure}

\section{Demonic Refinement Algebra}\label{S-general-algebra}

The basis for our program algebra is similar to von Wright's Demonic Refinement Algebra  (DRA) 
\cite{Wright04},
which is designed to support algebraic reasoning in a refinement calculus style
\cite{Bac81a,BackWright98,Morgan94,ATBfSRatPC}.
We define the following structure.
\vspace*{-0.5ex}
\begin{eqnarray*}
(\Commands, \Nondet, \bigsqcup, ~\Seq\, , \Nil)
\vspace*{-0.5ex}
\end{eqnarray*}
where the carrier set $\cal{C}$ is the set of \emph{commands}
containing distinguished element $\Nil$, the identity of the
sequential composition operator, ``$\,\Seq\,$''.
The operators $\Nondet$ and $\bigsqcup$ take sets of commands and return a command.
We use the notation $c \nondet d$ to stand for $\Nondet \{c,d\}$ and
$c \sqcup d$ for $\bigsqcup \{c,d\}$.
Unary operators have higher precedence than all binary operators.
Sequential composition has higher precedence than all other binary operators 
and non-deterministic choice has lower precedence than all other operators,
but otherwise we make no assumptions about precedence.

Commands form a complete distributive lattice 
$(\cal{C}, \Nondet, \bigsqcup, \bot, \top)$ 
with \emph{nondeterministic choice} as the lattice meet ($c\nondet d$),
and \emph{conjunction} of commands as the lattice join ($c \sqcup d$).  
The top of the lattice $\top \sdefs \Nondet \emptyset$ is the infeasible command 
(called ``magic'' in the refinement calculus) 
and the bottom of the lattice $\bot \sdefs \bigsqcup \emptyset$ is the command that aborts.  
The partial order defined on commands is the refinement relation $c \refsto d$
meaning $c$ is refined (or implemented by) $d$. 
For any commands $c,d \in \cal{C}$, 
$c \refsto d \sdefs (c \nondet d) = c$,
and hence $\bot \refsto c \refsto \top$. 
We refer to this as the \emph{refinement lattice} (see Figure~\ref{lattices}).
Note that because DRA is a \emph{refinement} algebra it uses $\refsto$ as its partial order 
instead of Kozen's $\geq$ and hence our lattice of commands is the dual of Kozen's lattice
for Kleene Algebra with Tests (KAT) \cite{kozen97kleene}, 
i.e., $\sqcap$ in DRA matches $\sqcup$ in KAT, and $\sqcup$ in DRA matches $\sqcap$ in KAT.
Given commands form a complete lattice, for any monotone function
least/greatest fixed points are well defined.  In particular, fixed
points are used to define iteration operators below.

\begin{figure}
\figurerule\\
Let $c$ and $d$ be commands and $C$ and $D$ be set of commands.
\\[2ex]
\begin{minipage}[t]{0.5\textwidth}
\textbf{Sequential}
\begin{eqnarray}
  c_0 \Seq (c_1 \Seq c_2) & = & (c_0 \Seq c_1) \Seq c_2 \label{A-seq-asso} \\
  c \Seq \Nil ~= &c&  =~  \Nil \Seq c \label{A-seq-identity} \\
   \bot \Seq c  &=& \bot \label{A-seq-annihilation-left}
\end{eqnarray}
\end{minipage}%
\begin{minipage}[t]{0.5\textwidth}
\begin{eqnarray}
  (\Nondet C) \Seq d & = & \Nondet_{c \in C} (c \Seq d)  \label{A-seq-distr-right}\\
  D \neq \emptyset ~\implies~ c \Seq (\Nondet D) & = & \Nondet_{d \in D}(c \Seq d) \label{A-seq-distr-left} 
\end{eqnarray}
\end{minipage}
\figurerule

\caption{Axioms for sequential composition}\label{figure:axioms}
\end{figure}

The axioms for sequential composition are given in Figure~\ref{figure:axioms}.
Sequential composition of commands ($c \,\Seq\, d$) is associative (\ref{A-seq-asso}) and has
identity $\Nil$ (\ref{A-seq-identity}).  
The least command, $\bot$,  is a left (but not right) annihilator (\ref{A-seq-annihilation-left}), and 
sequential composition distributes over arbitrary choices from the right (\ref{A-seq-distr-right}),
hence $\top$ is also an annihilator from the left,
i.e.\ $\top \Seq c = \top$.
Sequential composition distributes over non-empty choices from the left (\ref{A-seq-distr-left}), 
however, because $\Nondet \emptyset = \top$ and $\bot \Seq \top = \bot$ by (\ref{A-seq-annihilation-left}),
it does not distribute from the left over the empty choice, and so it is conjunctive, but not universally conjunctive. 
\footnote{Conjunctivity holds, for example, for our relational model in Section~\ref{S-rg-logic}.}

\subsection{Iteration of commands}\label{S-iteration}

The iteration of a command is inductively defined as 
$c^0 = \Nil$ and $c^{i+1} = c \Seq c^i$.
More general iteration operators are captured via greatest ($\nu$) and least
($\mu$) fixed points of the complete lattice: 
\begin{eqnarray}
\Fin{c} &\sdef& (\nu x .\Nil \nondet c \Seq x) \label{fin-iter-def}\\
\Om{c} &\sdef& (\mu x . \Nil \nondet c \Seq x) \label{om-iter-def}\\
\Inf{c} &\sdef& \Om{c} \Seq \top                    \label{inf-iter-def}
\end{eqnarray}
Finite iteration zero or more times is defined by (\ref{fin-iter-def}), and 
(\ref{om-iter-def}) defines finite or possibly infinite iteration.
Infinite iteration is defined through the possibly infinite iteration 
followed by $\top$ and hence excludes any finite number of iterations (\ref{inf-iter-def}). 
A number of useful laws can be derived.
The unfolding laws (\ref{L-omega-unfold})
and (\ref{L-finite-unfold})
and the induction laws (\ref{L-omega-induction}) 
and (\ref{L-finite-induction})
result from the fixed point definitions for iterations.
Law~(\ref{L-infinite-unfold}) follows from (\ref{L-omega-unfold}) and the definition of $\Inf{c}$
(\ref{inf-iter-def}),
which also justifies (\ref{L-infinite-annihilates}).
Law~(\ref{L-infinite-unfold-power}) follows from (\ref{L-infinite-unfold}) by induction on $i$.
Iterations are also idempotent under sequential composition (\ref{L-infinite-twice}) an (\ref{L-finite-twice}). \\
\noindent
\begin{minipage}[t]{0.5\textwidth}
\begin{eqnarray}
  \Om{c} &=& \Nil \nondet c \Seq \Om{c} \label{L-omega-unfold} \\
  \Fin{c} &=& \Nil \nondet c \Seq \Fin{c} \label{L-finite-unfold} \\
  \Inf{c} &=& c \Seq \Inf{c} \label{L-infinite-unfold}\\
  \Inf{c} &=& c^i \Seq \Inf{c} \label{L-infinite-unfold-power}\\
  \Inf{c} &=& \Inf{c}\Seq d \label{L-infinite-annihilates}
\end{eqnarray}
\end{minipage}%
\begin{minipage}[t]{0.5\textwidth}
\begin{eqnarray}
  \Nil \nondet c \Seq x \refsto x &\implies& \Om{c} \refsto x \label{L-omega-induction} \\
  x \refsto \Nil \nondet c \Seq x &\implies& x \refsto \Fin{c} \label{L-finite-induction} \\
  \Om{c}\Seq \Om{c} &=& \Om{c} \label{L-infinite-twice}\\
  \Fin{c}\Seq \Fin{c} &=& \Fin{c} \label{L-finite-twice}
\end{eqnarray}
\end{minipage}\vspace*{2ex}\\
Conjunctivity (\ref{A-seq-distr-left}) is required to show Laws~(\ref{L-isolation}) and (\ref{L-finite-iteration}).
Law (\ref{L-iteration1}) 
follows from (\ref{L-isolation}) and (\ref{L-infinite-annihilates}).
\\
\begin{minipage}[t]{0.5\textwidth}
\begin{eqnarray}
  \Om{c} &=& \Fin{c} \nondet \Inf{c} \label{L-isolation} \\
  \Fin{c} &=& \textstyle\Nondet_{i \in \nat} c^i\label{L-finite-iteration}
\end{eqnarray}
\end{minipage}
\begin{minipage}[t]{0.5\textwidth}
\begin{eqnarray}
  \Om{c} \Seq d &=& \Fin{c} \Seq d \nondet \Inf{c} \label{L-iteration1} 
\end{eqnarray}
\end{minipage}
\vspace*{1ex}

\subsection{The sub-algebra of tests}\label{S-tests}

Tests are special commands that are used to model conditionals and while-loops and hence form an essential construct when reasoning about programs in any state-based formalism.  Assume $t$ is a test, $\tnegate{t}$ is its negation, and $c$ and $d$ are commands,  an abstract algebraic representation of conditionals and while-loops for sequential programs is given by
\vspace*{-1ex}
\begin{eqnarray*}
  \text{\bf if }  t ~\text{\bf then } c ~\text{\bf else } d \sdefs t \Seq c \nondet \tnegate{t} \Seq d
  ~~~\text{ and }~~~
  \text{\bf while } t ~\text{\bf do } c \sdefs \Om{(t \Seq c)} \Seq \tnegate{t}~.
\vspace*{-1ex}
\end{eqnarray*}
Blikle \cite{Blikle78} used this style of representation of programs in a relational algebra  and Gardner and Morgan \cite{gardinerfacj93} and von Wright \cite{Wright04} in the refinement calculus. Kozen \cite{kozen97kleene} provided \emph{Kleene Algebra with Tests} (KAT) as an abstract-algebraic framework for reasoning about programs with tests. 
In Kozen's approach, tests form a Boolean sub-algebra within Kleene algebra.
We follow his construction here. To introduce a sub-lattice of test commands (see Figure~\ref{lattices}) to a DRA
$(\Commands, \Nondet, \bigsqcup,  ~\Seq\, , \Nil)$, 
we augment it with an additional carrier set $\Tests \subseteq \Commands$ of \emph{test commands} and a unary operation of complementation on tests, $\tnegate$, such that
\vspace*{-1ex}
\begin{eqnarray*}
  (\Tests, \nondet, \sqcup, ~\tnegate\, , \top, \Nil)
\vspace*{-1ex}
\end{eqnarray*}
is a Boolean algebra with least element $\Nil$ (a succeeding test) 
and greatest element $\top$ (a failing test). 
Additionally, we require tests $t, t' \in \Tests$ to satisfy the following interchange axiom:
\begin{eqnarray}
(t \Seq c) \sqcup (t' \Seq d) &=&  (t \sqcup t') \Seq (c \sqcup d) \label{A-test-sup-interchange}
\end{eqnarray}
from which we can show that
\begin{eqnarray}
     t \Seq t'
 =  (t \sqcup \Nil) \Seq (\Nil \sqcup t')
 =  (t \Seq \Nil) \sqcup (\Nil \Seq t')
&=&  t \sqcup t'      \label{test-seq-test}
\end{eqnarray}
giving us closure of tests over sequential composition.
Given a program intuition, the greatest element in the sub-lattice of tests, $\top$, corresponds to the test that always evaluates to false, and the least element, $\Nil$, corresponds to the test that always evaluates to true.
The smallest non-trivial Boolean sub-algebra of tests is $\{\top, \Nil\}$, and is the interpretation of tests for event-based formalisms without an explicit notion of state. 

\subsubsection{Assertions}\label{S-assertions}

Tests also give rise to the concept of \emph{assertions} (preconditions) \cite{Wright04,Solin07}. The assertion corresponding to a test $t$ is a command which terminates if the test holds and aborts if the test does not hold,
i.e., 
\vspace*{-1ex}
\begin{eqnarray*}
\Assert{t} = t \nondet \tnegate{t} \Seq \bot~. 
\vspace*{-1ex}
\end{eqnarray*}

Note that KAT is not rich enough to model assertions as it does not include abort.
Von Wright \cite{Wright04} has shown that assertions and tests form a Galois connection as follows,
a law that can be proved from our axiomatisation:
\vspace{-1ex}
\begin{eqnarray*}
  (\Assert{t}) \Seq c \refsto d & \iff & c \refsto t \Seq d~.
\vspace*{-1ex}
\end{eqnarray*}

\subsubsection{Encoding in Isabelle}\label{S-Boolean-encoding}

Armstrong et al.\ \cite{Armstrong14} present three reference formalisations of Kozen's KAT in the theorem prover Isabelle, of which we follow the two-sorted implementation.  
Tests are embedded into the lattice of commands via a mapping, an injective homomorphism, from a Boolean algebra $(\BooleanAlg, \union, \inter, \bar{~~}, \bot, \top)$, into the Boolean sub-algebra of test commands:
\begin{eqnarray*}
 \tau : \BooleanAlg \rightarrow \Commands
\end{eqnarray*}
giving us the following correspondence between the operations of the Boolean algebra ($\BooleanAlg$) and tests ($\Tests$):\\
\begin{minipage}{0.5\textwidth}
\begin{eqnarray}
  \tau(p \union q) &=& \tau(p) \nondet \tau(q)        \label{tau-inf-tau}\\
  \tau(p \inter q) &=& \tau(p) \sqcup \tau(q)         \label{tau-sup-tau} \\
  \tau(\bar{p}) &=& \tnegate\, \tau(p)                \label{tau-negate}\\
  p \subseteq q  &\iff & \tau(q) \refsto \tau(p)        \label{tau-ordering}
\end{eqnarray}
\end{minipage}%
\begin{minipage}{0.5\textwidth}
\begin{eqnarray}
  \tau(\bot)  &=& \top                             \label{tau-bot-top}\\
  \tau(\top)  &=& \Nil                             \label{tau-top-nil}\\
  \tau(p \inter q) &=& \tau(p) \Seq \tau(q)        \label{tau-seq-tau}
\end{eqnarray}
\end{minipage}\vspace{1ex}\\
Since the natural ordering on the Boolean algebra, $p \subseteq q \iff p \union q = q$, is the inverse of our ordering on commands, $c \refsto d \iff c \nondet d = c$, we have that the partial ordering on the Boolean algebra maps to the inverse of the refinement ordering on tests (\ref{tau-ordering}), 
and the bottom element of the Boolean algebra ($\bot$) is mapped to the top element of tests ($\top$), etc. 
Property (\ref{tau-seq-tau}) follows from $(\ref{test-seq-test})$ and $(\ref{tau-sup-tau})$.
With this encoding, Isabelle's Boolean algebra theory can readily be applied to test commands.

\subsection{The sub-algebra of atomic steps}\label{S-atomic}

This section introduces the sub-algebra of commands that correspond to atomic steps, $\AtomicSteps \subseteq \cal{C}$. These atomic steps will be later shown to have both an interpretation that corresponds to Aczel's program and environment steps (Section~\ref{S-rg-logic}), and to events from CCS an CSP (Section~\ref{S-process-algebras}). In the following, the term \emph{step} is used exclusively for an atomic step. 

In the same manner that tests form a sub-lattice of commands,  the set of atomic steps forms a sub-lattice of commands which is a Boolean algebra (see Figure~\ref{lattices}). In particular, we can extend a DRA,
$(\Commands, \Nondet, \bigsqcup,  ~\Seq\, , \Nil)$,
with another carrier set $\AtomicSteps \subseteq \Commands$, negation operator, $\anegate$, and distinguished element $\cstepd$ such that 
\begin{eqnarray*}
  (\AtomicSteps, \nondet, \sqcup,\, \anegate\, , \top, \, \cstepd)
\end{eqnarray*}
forms a Boolean algebra, and atomic steps $a,b\in \AtomicSteps$ satisfy the interchange axiom:
\begin{eqnarray}
(a \Seq c) \sqcup (b \Seq d) &=&  (a \sqcup b) \Seq (c \sqcup d) \label{A-step-sup-interchange}
\end{eqnarray}
However, unlike for tests, this does not imply that atomic steps are closed under sequential composition, because $\cstepd$ is not defined to be $\Nil$, the identity of sequential composition.

As for commands, the meet corresponds to non-deterministic choice, $a \nondet b$, and can behave as either $a$ or $b$. The join of two steps, $a \sqcup b$, can be thought of as a step that both $a$ and $b$ agree to do.
Distinguished element $\cstepd$ is the least step in the refinement lattice, i.e. the step such that $a \nondet \cstepd = \cstepd$ for any atomic step $a$. It corresponds to the non-deterministic choice over all possible atomic steps. The greatest step in the refinement lattice, $\top$ (also the greatest command in $\Commands$), can be thought of as a step that cannot be taken because it is infeasible. From the axioms of Boolean algebra we have that
\begin{eqnarray*}
a \sqcup \anegate a  =  \top &\textrm{~and~}& a \nondet \anegate a = \cstepd 
\end{eqnarray*}
which represents the fact that steps $a$ and $\anegate{a}$ have no common behaviour, and that $\anegate{a}$ has all the step behaviours that $a$ does not have.
Negation for tests ($\tnegate$) differs from negation for atomic steps ($\anegate$) 
because we have $\tnegate \top = \Nil$ but $\anegate \top = \cstepd$.

\subsubsection{Assumptions}\label{S-assumptions}

The inclusion of a negation operator on steps allows one to define an equivalent of an assertion (see Section~\ref{S-tests}) for atomic steps on the abstract level.
For any step $a$ define,
\begin{eqnarray}
\Assume{a} \sdef a \nondet \anegate a \Seq \bot~. \label{def-assume}
\end{eqnarray}
The command $\Assume{a}$ behaves as $a$ and terminates, or as $\anegate{a}$ and aborts.
It represents an assumption that step $a$ occurs in the sense that any other step violates the assumption, leading to a state in which any other behaviour is then possible.
It becomes a useful tool when reasoning with rely conditions 
which specify assumptions about the environment's behaviour (see Section~\ref{SS-rely-guar}).  
We have that conjunction of two assumptions may be simplified as follows.
\begin{lemma}[conjoin-assumptions]\label{L-conjoin-assume}
For $a,b \in \AtomicSteps$, we have that
$\Assume{a} \sqcup \Assume{b}  =  \Assume{(a \nondet b)}$~.
\end{lemma}

\begin{proof}
\begin{displaymath}
\Assume{a} \sqcup \Assume{b}
\Equals
(a \nondet \anegate a \Seq \bot) \sqcup (b \nondet \anegate b\Seq \bot)
\Equals*[distribute conjunction over choices, $\Nil$ is the identity of sequential composition (\ref{A-seq-identity})]
(a \sqcup b) \nondet
(a\Seq \Nil \sqcup \anegate b \Seq \bot) \nondet
(\anegate a \Seq \bot \sqcup b \Seq \Nil) \nondet
(\anegate a\Seq \bot \sqcup \anegate b\Seq \bot)
\Equals*[interchange atomic steps over conjunction (\ref{A-step-sup-interchange})]
(a \sqcup b) \nondet
(a \sqcup \anegate b)\Seq (\Nil \sqcup \bot) \nondet
(\anegate a \sqcup b)\Seq (\bot \sqcup \Nil) \nondet
(\anegate a \sqcup \anegate b)\Seq \bot
\Equals*[use $\Nil\sqcup\bot = \Nil$, and $\Nil$ is the identity of sequential composition (\ref{A-seq-identity})]
(a \sqcup b) \nondet
(a \sqcup \anegate b) \nondet
(\anegate a \sqcup b) \nondet
(\anegate a \sqcup \anegate b)\Seq \bot
\Equals*[simplify using the Boolean algebra of atomic steps]
(a \nondet b) \nondet
\anegate (a\nondet b) \Seq \bot
\Equals
\Assume{(a \nondet b)}
\end{displaymath}
\end{proof}

From which we have that assumptions are anti-monotonic in their argument, corresponding to the intuition that we can weaken an assumption in a refinement step. 
\begin{lemma}[weaken-assume]\label{L-weaken-assume}
For $a,b \in \AtomicSteps$, if $b \refsto a$ then $\Assume{a} \refsto \Assume{b}$.
\end{lemma}

\begin{proof}
From the assumption, $a \nondet b = b$, hence
\(
  \Assume{a} \refsto \Assume{a} \sqcup \Assume{b} = \Assume{(a \nondet b)} = \Assume{b}~
\)
\end{proof}

\begin{lemma}[iterated-assumption]\label{L-iterated-assumption}
For $a \in \AtomicSteps$,
\(
\Om{(\Assume{a})}  =  \Om{a} \Seq (\Nil \nondet \anegate a \Seq \bot)~.  
\)
\end{lemma}

\begin{proof}
The iteration can be simplified using the decomposition lemma of DRA, 
i.e. $\Om{(c \sqcap d)} = \Om{c} \Seq \Om{(d \Seq \Om{c})}$.
\[
  \Om{(\Assume{a})}
 =
  \Om{(a \nondet \anegate a \Seq \bot)}
 = 
  \Om{a} \Seq \Om{(\anegate a \Seq \bot \Seq \Om{a})}
 =
  \Om{a} \Seq \Om{(\anegate a \Seq \bot)}
 =
  \Om{a} \Seq (\Nil \nondet \anegate a \Seq \bot)
\]
\end{proof}

\subsubsection{Encoding in Isabelle}

Our encoding of the sub-algebra of atomic steps in Isabelle follows the same format as that of the Boolean Algebra of tests presented in Section~\ref{S-Boolean-encoding}, and is achieved using an injective homomorphism from some Boolean algebra, $(\BooleanAlg, \union, \inter, \bar{~~}, \bot, \top)$, (typically not the same Boolean algebra used for tests) into the Boolean sub-algebra of atomic steps, $\cstepd: \BooleanAlg \rightarrow \Commands$.

\subsection{Atomic steps and tests}\label{S-atomictests}

Extending a DRA with both atomic steps and tests gives us a three-sorted algebra
\begin{eqnarray*}
(\Commands, \Tests, \AtomicSteps, \Nondet, \bigsqcup,  ~\Seq\, , ~\tnegate\, , \anegate\, , \Nil, \cstepd)
\end{eqnarray*}
to which we add additional axioms to describe the interactions of steps and tests. As in the Figure~\ref{lattices},
we assume that tests and atomic steps share only one element ($\top$) and hence include
\vspace*{-1ex}
\begin{eqnarray}
 \cstepd \sqcup \Nil = \top \label{A-test-atomic-sup}
\vspace*{-3ex}
\end{eqnarray}
in our axiomatisation. We also take as an axiom that an atomic step preceded by any test $t$ is also a step:
\vspace*{-1ex}
\begin{eqnarray}
t \Seq a &\in & \AtomicSteps \label{A-test-atomic-pre}
\vspace*{-3ex}
\end{eqnarray}
although we make no assumption about the succession of a step by a test.

\section{Synchronous Refinement Algebra}\label{S-sync-algebra}

\begin{figure}
\figurerule\\
For any commands $c, d \in \Commands$, set of commands $D \subseteq \Commands$, 
atomic steps $a, b \in \AtomicSteps$, and test $t \in \Tests$,
\\[-2ex]
\begin{minipage}[t]{0.50\textwidth}
\begin{eqnarray}
c_0 \sync (c_1 \sync c_2) &=& (c_0 \sync c_1) \sync c_2 \label{A-sync-assoc}\\
c \sync d &=& d \sync c \label{A-sync-comm}\\
c \sync \syncidcommand &=& c \label{A-sync-id-command} \\
D \neq \emptyset ~\implies~ c \sync (\Nondet D) &=& (\Nondet_{d \in D} c \sync d) \label{A-sync-Inf-distrib}\\
a \sync b &\in& \AtomicSteps \label{A-sync-closed}\\
a \sync ~{\syncid}   &=& a \label{A-sync-id} 
\end{eqnarray}
\end{minipage}%
\begin{minipage}[t]{0.50\textwidth}
\begin{eqnarray}
c_0 \Seq c_1 \sync d_0 \Seq d_1 &\refsto& (c_0 \sync d_0) \Seq (c_1 \sync d_1) \label{A-sync-weak-interchange-seq} \\ 
a \Seq c \sync b \Seq d &=& (a \sync b) \Seq (c \sync d) \label{A-atomic-sync-interchange}\\
\Inf{a} \sync \Inf{b} & = & \Inf{(a \sync b)}  \label{A-atomic-infiter-sync}\\
\Nil \sync \Nil &=& \Nil  \label{A-nil-sync-absorb}\\
a \Seq c \sync \Nil &=& \top \label{A-nil-sync-atomic} \\
t \Seq c \sync t \Seq d &=& t \Seq (c \sync d) \label{A-test-sync-interchange}
\end{eqnarray}

\end{minipage}\\[1ex]
\figurerule
\caption{Axioms for the synchronisation operator}\label{figure:axioms-sync}
\end{figure}

This section adds an abstract synchronisation operator, $\sync$, to the algebra
to form a Synchronous Refinement Algebra (SRA). 
We take the DRA with atomic steps and tests (Section~\ref{S-general-algebra}) 
and introduce a synchronisation operation $\sync$ 
and distinguished elements $Id_{}\in \Commands$ and $\syncid \in \AtomicSteps$. 
The axiomatisation is sufficiently general so as to allow multiple interpretations of the operation, 
e.g. as parallel composition and weak conjunction (see Section~\ref{S-instantiations}).

Figure~\ref{figure:axioms-sync} gives the axioms for the abstract synchronisation operator. 
First, we have that $(\Commands, \sync, \syncidcommand)$ is a commutative monoid over commands. 
That is, $\sync$ is associative (\ref{A-sync-assoc}), commutative (\ref{A-sync-comm}) and 
has identity $\syncidcommand$ (\ref{A-sync-id-command}). 
Note that the identity typically differs for the different instantiations of $\sync$, 
e.g.\ as parallel or weak conjunction in Section~\ref{S-instantiations}.
The operator $\sync$ distributes over non-deterministic choices over non-empty sets of commands (\ref{A-sync-Inf-distrib}), 
and so it is monotonic. 
For an empty set of commands $\Nondet \emptyset = \top$, 
but we do not have a law of the form $c \sync \top = \top$ 
because operators like parallel are abort strict, i.e.~$\bot \parallel d = \bot$, 
where $\bot$ represents the command abort, 
the bottom of the refinement lattice and hence $\bot \parallel \top = \bot$.

Atomic steps are closed under the synchronisation operator (\ref{A-sync-closed}) and have atomic-step identity $\syncid \in \AtomicSteps$ (\ref{A-sync-id}). From these properties, we have for example that for any atomic step $a \in \AtomicSteps$
\begin{eqnarray}
a \sync \cstepd
= a \sync (\syncid \sqcap \anegate \syncid)
= a \sync \syncid \sqcap a \sync \anegate \syncid 
& \refsto &  a \sync \syncid  =  a ~,\label{A-sync-cstepd}
\end{eqnarray}
and so taking $a$ to be $\cstepd$ gives $\cstepd \sync \cstepd = \cstepd$
because $\cstepd$ is the least atomic step command. 

For arbitrary commands, operator $\sync$ and sequential composition 
satisfy a weak interchange axiom (\ref{A-sync-weak-interchange-seq}).
A sequence of commands $c_0 \Seq c_1$ may synchronise with a sequence $d_0 \Seq d_1$
by synchronising $c_0$ with $d_0$ and then synchronise $c_1$ with $d_1$.
The axiom is only a refinement because synchronising $c_0$ and $d_0$
(or $c_1$ and $d_1$) may be infeasible,
whereas on the left $c_0$ may synchronise with the whole of $d_0$ and part of $d_1$
and $c_1$ with the rest of $d_1$, or vice versa.
Axiom (\ref{A-atomic-sync-interchange}) describes 
how the synchronisation of complex commands decomposes into 
the synchronisation of the atomic steps that constitute it:  
two commands that both have leading atomic steps are synchronised by 
synchronising the leading atomic steps, 
followed by the synchronisation of the remainder of the commands. 
If $a$ and $b$ cannot synchronise then $a \sync b$ is infeasible ($\top$). 
We take synchronisation of infinite iterations of atomic steps as an axiom (\ref{A-atomic-infiter-sync}), although independence of that axiom from the others is an open question.

The distinguished test $\Nil$, the identity of sequential composition, terminates immediately without performing any atomic steps at all, and so it synchronises with itself (\ref{A-nil-sync-absorb}), but when synchronised with a process that must perform an atomic step before terminating,  their composition is infeasible (\ref{A-nil-sync-atomic}). Since synchronisation is monotonic, we have from (\ref{A-nil-sync-atomic}) that for any arbitrary test $t\in \Tests$,
\begin{eqnarray}
a\Seq c \sync t & = & \top \label{A-sync-test} ~,
\end{eqnarray}
and more specifically
\begin{eqnarray}
a\Seq c \sync \top & = & \top \label{A-sync-top}
\end{eqnarray}
because $\Nil \refsto t \refsto \top$.
Axiom (\ref{A-test-sync-interchange}) defines how tests distribute over synchronisation for arbitrary commands. Although we \emph{do not} in general have the stronger axiom $c \sync t\Seq d = t \Seq (c \sync d)$, we do have that the property holds, for example, when
(i) $c$ is a test,
(ii) $c$ is preceded by an atomic step, or
(iii) $c$ is a refinement of $\syncidcommand$.
\begin{lemma}[test-command-sync-command] \label{L-test-command-sync-command}
For any atomic command $a \in \AtomicSteps$, commands $c, c', d \in \Commands$,  and tests $t,t' \in \Tests$, we have that if either $c = t'$, or $c = a\Seq c'$,  or $\syncidcommand \refsto c$, then
\begin{eqnarray*}
c \sync t\Seq d  &=& t \Seq (c \sync d)
\end{eqnarray*}
holds.
\end{lemma}

\begin{proof}
First we have that
\begin{displaymath}
t \Seq ( c \sync t\Seq d) 
=  (t\Seq c \sync t\Seq t \Seq d)
=  ( t \Seq c \sync t\Seq d) 
= t \Seq (c \sync d) 
\end{displaymath}
by distributing tests over synchronisation (\ref{A-test-sync-interchange}), simplifying sequential composition of tests using (\ref{test-seq-test}), and distributing the test back over the synchronisation. 
Next, we have
\begin{displaymath}
\tnegate t \Seq ( c \sync t\Seq d) 
\Equals*[distribute test over synchronisation (\ref{A-test-sync-interchange})]
(\tnegate t \Seq c \sync \tnegate t \Seq t \Seq d) 
\Equals*[$\tnegate t \Seq t = \top$ from (\ref{test-seq-test}) and $\top\Seq d = \top$ from (\ref{A-seq-distr-right})]
\tnegate t \Seq c \sync \top
\end{displaymath}
which simplifies to $\top$ using either
(\ref{A-sync-test}) if $c = t'$,
(\ref{A-sync-top}) if $c = a \Seq d'$,
and  $\syncidcommand \sync \top  = \top$ (\ref{A-sync-id}) and monotonicity of synchronisation (\ref{A-sync-Inf-distrib}) if $\syncidcommand \refsto c$.

Using $\Nil$ is the identity of sequential composition (\ref{A-seq-identity}), distribution of choice (\ref{A-seq-distr-right}) and the above results we finally have:
\begin{displaymath}
  c \sync t \Seq d
= \Nil \Seq (c \sync t \Seq d)
= (t \nondet \tnegate t) \Seq (c \sync t \Seq d)
= t \Seq (c \sync t\Seq d) \nondet \tnegate t \Seq (c \sync t \Seq d)
= t \Seq (c \sync d) ~.
\end{displaymath}
\end{proof}

From this, we can also show that the synchronisation operator behaves like a conjunction on tests. 
Hence tests, like atomic steps, are closed under synchronisation.
\begin{lemma}[test-sync-test] \label{L-test-sync-test}
For tests $t, t' \in \Tests$, we have
$t \sync t' = t \sqcup t'$.
\end{lemma}

\begin{proof}  
Since $t \in \Tests$, we can use Lemma~\ref{L-test-command-sync-command} to distribute the tests over the synchronisation to the left. Applying (\ref{A-seq-identity}), Lemma~\ref{L-test-command-sync-command} twice, (\ref{A-nil-sync-absorb}) and (\ref{test-seq-test}), gives
\begin{displaymath}
   t \sync t'
=  (t\Seq \Nil) \sync (t' \Seq \Nil)
=  t\Seq t' \Seq (\Nil \sync \Nil)
= (t \Seq t) \Seq \Nil
=  t \sqcup t' ~.
\end{displaymath}
\end{proof}

The synchronisation operator does not in general distribute over sequential composition
but the weak interchange axiom (\ref{A-sync-weak-interchange-seq}) can be used to show
a form of distribution.
\begin{lemma}[sync-distribute-seq]\label{L-sync-distribute-seq}
If $c \refsto c \Seq c$, then
$c \sync (d_0 \Seq d_1) \refsto (c \sync d_0) \Seq (c \sync d_1)$.
\end{lemma} 
\begin{proof}
Using (\ref{A-sync-weak-interchange-seq}),
$c \sync (d_0 \Seq d_1) \refsto (c \Seq c) \sync (d_0 \Seq d_1) \refsto (c \sync d_0) \Seq (c \sync d_1)$.
\end{proof}

\section{Properties of iterations of atomic steps}\label{S-atomic-iteration}

In addition to defining programming statements such as while loops, 
iterators are used to build specifications from atomic steps.  
For instance, commands corresponding to Jones' rely and guarantee concepts are 
constructed as iterations of relatively straightforward commands 
that make assumptions about the steps of the environment and
constrain the steps of the program, respectively
(see Section~\ref{S-rely-guarantee}).  Below we provide
some useful properties of synchronisation over atomic iterations.

Isabelle/HOL proofs of these lemmas have been completed and 
they may be also found in Appendix~\ref{S-proofs} for review.
Note that all properties in this section are proven on the level of the 
(abstract) synchronisation operator and hence hold for any instantiations,
e.g.\ $\parallel$ and $\together$. 
The use of the abstract operator $\sync$ helps to highlight
these as properties that are shared by all synchronisation operators.

Because $\Nil$ performs no steps, if it synchronises with a (possibly) finite iteration,
the composition cannot perform any steps but can terminate and hence equals $\Nil$.
If $\Nil$ synchronises with an infinite iteration, 
the combination cannot perform any steps but cannot terminate,
and hence equals the infeasible command $\top$.
\begin{lemma}[atomic-iteration-nil]\label{L-atomic-iteration-nil}
Let $a$ be an atomic command.
\begin{eqnarray*}
  \Fin{a}  \sync \Nil = \Nil
\hspace*{1cm}
  \Om{a} \sync \Nil = \Nil
\hspace*{1cm}
  \Inf{a} \sync \Nil = \top
\end{eqnarray*}
\end{lemma}

\begin{proof}
The properties follow from axioms (\ref{A-nil-sync-absorb}) and (\ref{A-nil-sync-atomic}) 
using unfolding of the iterations:
$\Fin{a} = \Nil \nondet a \Seq \Fin{a}$
and
$\Om{a} = \Nil \nondet a \Seq \Om{a}$
and
$\Inf{a} = a \Seq \Inf{a}$.
\end{proof}

For the following lemmas, 
let $a$ and $b$ be atomic steps, and $c$ and $d$ any commands.
Axiom (\ref{A-atomic-sync-interchange}) can be extended to iteration $i$ times 
as given in the following lemma, which is proven by induction on $i$.
\begin{lemma}[atomic-iteration-power]\label{L-atomic-iteration-power}
\(
 a^i \Seq c \sync b^i \Seq d  ~=~  (a\sync b)^i \Seq (c \sync d)
\)
\end{lemma}
Choosing $c$ and $d$ to both be $\Nil$ gives the corollary that $a^i \sync b^i = (a \sync b)^i$.

In the following lemmas we then generalize Lemma~\ref{L-atomic-iteration-power} further to account for cases where atomic steps $a$ and $b$ are iterated an arbitrary number of times before pre-composing them with $c$ and $d$, respectively. For these cases, we take into consideration situations
where there may be more iterations of $a$ than $b$ (and hence the additional iterations of $a$ are in parallel with the start of $d$),
or the symmetric case when there may be more occurrences of $b$ than $a$. 
The proofs of these lemmas rely on the conjunctivity axiom (\ref{A-seq-distr-left}), and the properties (\ref{L-isolation}) and (\ref{L-finite-iteration}) that are derived from it. 
First we consider the case where $a$ and $b$ are iterated an arbitrary finite number of times.
\begin{lemma}[atomic-iteration-finite]\label{L-atomic-iteration-finite}
\(
 \Fin{a} \Seq c \sync \Fin{b} \Seq d 
   = \Fin{(a \sync b)} \Seq ((c \sync \Fin{b} \Seq d) \nondet (\Fin{a} \Seq c \sync d))
\)
\end{lemma}
The proof of this lemma can be found in Appendix~\ref{S-proofs}, on page \pageref{P-atomic-iteration-finite}.
Unfolding $\Fin{b}$ and $\Fin{a}$ using (\ref{L-finite-unfold}) gives the following corollary.
\begin{corollary}[atomic-iteration-finite]\label{C-atomic-iteration-finite}
\(
 \Fin{a} \Seq c \sync \Fin{b} \Seq d 
  = \Fin{(a \sync b)} \Seq ((c \sync d) \nondet (c \sync b \Seq \Fin{b} \Seq d) \nondet (a \Seq \Fin{a} \Seq c \sync d))
\)
\end{corollary}
Choosing $c$ and $d$ to both be $\Nil$ gives
the following as a corollary using (\ref{A-nil-sync-absorb}) and (\ref{A-nil-sync-atomic}).
\begin{corollary}[atomic-finite-sync]\label{L-atomic-finite-sync}
$\Fin{a} \sync \Fin{b} = \Fin{(a \sync b)}$
\end{corollary}

When $b$ is iterated an infinite number of times, there cannot be more iterations of $a$ than $b$, and so we get the following. 
\begin{lemma}[atomic-iteration-finite-infinite]\label{L-atomic-iteration-finite-infinite}
\( \Fin{a} \Seq c \sync \Inf{b} = \Fin{(a \sync b)} \Seq (c \sync \Inf{b})
\)
\end{lemma}
Appendix~\ref{S-proofs} provides the proof on page \pageref{P-atomic-iteration-finite-infinite}.
Combining Lemmas~\ref{L-atomic-iteration-finite} and \ref{L-atomic-iteration-finite-infinite} 
and using (\ref{L-iteration1}) the following lemma can be derived.
\begin{lemma}[atomic-iteration-finite-omega]\label{L-atomic-iteration-finite-omega}
$\Fin{a} \Seq c \sync \Om{b} \Seq d  =  \Fin{(a \sync b)} \Seq ((c \sync \Om{b} \Seq d) \nondet \Fin{a} \Seq c \sync d))$ 
\end{lemma}

The case in which $a$ and $b$ may both be iterated a finite or infinite number of times has the same structure as Lemma~\ref{L-atomic-iteration-finite}.
\begin{lemma}[atomic-iteration-either]\label{L-atomic-iteration-either}
\(
 \Om{a} \Seq c \sync \Om{b} \Seq d =
    \Om{(a \sync b)} \Seq ((c \sync \Om{b} \Seq d) \nondet (\Om{a} \Seq c \sync d))
\)
\end{lemma}
The proof is listed in Appendix~\ref{S-proofs} on page \pageref{P-atomic-iteration-either}.
Unfolding $\Om{a}$ and $\Om{b}$ using (\ref{L-omega-unfold}) gives the following corollary.
\begin{corollary}[atomic-iteration-either]\label{C-atomic-iteration-either}
\(
 \Om{a} \Seq c \sync \Om{b} \Seq d = 
    \Om{(a \sync b)} \Seq ((c \sync d) \nondet (c \sync b \Seq \Om{b} \Seq d) \nondet (a \Seq \Om{a} \Seq c \sync d))
\)
\end{corollary}
Choosing $c$ and $d$ to both be $\Nil$ gives the following corollary
using (\ref{A-nil-sync-absorb}) and (\ref{A-nil-sync-atomic}).
\begin{corollary}[atomic-either-sync]\label{L-atomic-either-sync}
\label{omega-sync-omega}
$\Om{a} \sync \Om{b} = \Om{(a \sync b)}$
\end{corollary}

In Section~\ref{S-rely} a rely command is defined in terms of an iteration 
that may abort if the environment performs a step not satisfying the rely condition.
Here we provide an abstract lemma for reasoning about an iteration of a command that
aborts after an $a_1$ step (representing the rely not holding).
Both weak conjunction and parallel are abort strict 
and hence satisfy the assumption of this lemma.
\begin{lemma}[iterations-with-abort]\label{L-iterations-with-abort}
Provided $c \sync \bot = \bot$ for all commands $c$,
\[
  \Om{(a_0 \nondet a_1 \Seq \bot)} \sync \Om{b} = \Om{(a_0 \sync b)} \Seq (\Nil \nondet (a_1 \sync b) \Seq \bot)~.
\]
\end{lemma}
The proof is listed in Appendix~\ref{S-proofs} on page \pageref{P-iterations-with-abort}.

\section{Instantiating parallel and weak conjunction as synchronous operators}\label{S-instantiations}

Given a DRA with atomic steps and tests as in Section~\ref{S-general-algebra}, 
we have that the conjunction operator, $\sqcup$, 
with synchronisation identity $\bot$
and atomic-step synchronisation identity $\cstepd$
satisfies the axioms of the synchronisation operator $\sync$ of SRA. 
In this section we give a further instantiation of synchronisation as two new operators: parallel and weak conjunction. 
We define a SRA with parallel ($\parallel$) and weak conjunction ($\together$) to be the three-sorted algebra
\begin{eqnarray*}
(\Commands, \Tests, \AtomicSteps, \Nondet, \bigsqcup,  ~\Seq\,, \parallel, \together, ~\tnegate\, , \anegate\, , \Nil, \cstepd, \Skip, \Chaos, \atomid)
\end{eqnarray*}
such that
\begin{eqnarray*}
(\Commands, \Tests, \AtomicSteps, \Nondet, \bigsqcup,  ~\Seq\,,  ~\tnegate\, , \anegate\, , \Nil, \cstepd)
\end{eqnarray*}
is a DRA with atomic steps and tests, and 
$(\parallel, \Skip, \atomid)$
is a synchronisation operator parallel, $\parallel$, with 
$\syncid$ 
taken to be the new distinguished element $\atomid\in \AtomicSteps$, and 
$\Identitycommand{}$
identified to be $\Skip = \Om{\atomid}$; and 
$(\together, \Chaos, \cstepd)$
is synchronisation operator weak conjunction, $\together$, with 
$\Identity{}$
taken to be $\cstepd$ and 
$\Identitycommand{}$
taken to be $\Chaos = \Om{\cstepd}$.
Further, we constrain $\together$ to be idempotent, 
\begin{eqnarray}
  c \together c & = & c \label{A-together-idempotent}
\end{eqnarray}
and both of the synchronisation operators parallel and weak conjunction to be abort-strict,
  \begin{eqnarray}
    c \parallel \bot & = & \bot  \label{A-parallel-abort}\\
    c \together \bot & = & \bot  \label{A-together-abort} 
  \end{eqnarray}
as well as introducing an additional interchange axiom to describe the interplay of the operators:
\begin{eqnarray}
(c_0 \parallel d_0) \together (c_1 \parallel d_1) &\refsto& (c_0 \together c_1) \parallel (d_0 \together d_1)~.
\label{conjunction-interchange-parallel}
\end{eqnarray}
That axiom codifies that one way of synchronising all the steps of $c_0 \parallel d_0$ and $c_1 \parallel d_1$
is to synchronise all the steps of $c_0$ with $c_1$ and synchronise all the steps of $d_0$ with $d_1$.

\subsection{Weak conjunction}

Weak conjunction behaves like conjunction ($\sqcup$) up until the failure of either command. It is a useful operator for composing a command together with a requirement that need only hold until the environmental assumptions of that command are violated, resulting in program failure (abortion). 

From idempotence (\ref{A-together-idempotent})  and monotonicity (\ref{A-sync-Inf-distrib}), we can show that the weak conjunction of any two commands $c$ and $d$, is refined by their conjunction:
\begin{eqnarray}
  c \together d & \refsto &  (c \sqcup d) \together (c \sqcup d) = c \sqcup d 
\end{eqnarray}
and that weak conjunction of any two commands $c$ and $d$ is equal to the conjunction of those commands if both $c$ and $d$ are refinements of the identity of weak conjunction (i.e.\ $\Chaos$), and hence they do not fail:
\begin{eqnarray}
\Chaos \refsto c \land \Chaos \refsto d & ~~\implies~~ &   c \together d = c \sqcup d  \label{A-weakconj-nonfail}
\end{eqnarray}
because
$c \sqcup d = (c \together \Chaos) \sqcup (\Chaos \together d) \refsto (c \together d) \sqcup (c \together d) = (c\together d)$.
This implies that weak conjunction behaves like conjunction on atomic steps $a,b\in \AtomicSteps$,
\begin{eqnarray}
  a \together b & = & a \sqcup b \label{A-together-atomic}~,
\end{eqnarray}
because for any atomic step $a$, $\Chaos = \cstepd^\omega \refsto \cstepd \refsto a$.
From the axioms of synchronisation and Lemma~\ref{L-test-sync-test} we already have that weak conjunction behaves like conjunction for tests, i.e. $t \together t' = t \sqcup t'$.
Using these properties we can, for example, show that for atomic steps $a_i,b_i\in \AtomicSteps$,
\begin{displaymath}
(a_1\Seq a_2 \Seq a_3 \Seq a_4) \together (b_1 \Seq b_2 \Seq \bot)
\Equals*[apply atomic step interchange axiom (\ref{A-atomic-sync-interchange}) twice]
(a_1 \together b_1) \Seq (a_2 \together b_2) \Seq ((a_3\Seq a_4) \together \bot)
\Equals*[weak conjunction is conjunction for atomic steps (\ref{A-together-atomic}) and is abort-strict (\ref{A-together-abort})]
(a_1 \sqcup b_1) \Seq (a_2 \sqcup b_2) \Seq \bot ~.
\end{displaymath}

The weak conjunction operators also simplifies to conjunction for  iterated assumptions, a fact that will be useful for proving properties on relies.
\begin{lemma}[assume-iter-conj-assume-iter]\label{L-assume-iter-conj-assume-iter}
For any atomic steps $a, b \in \AtomicSteps$,
\begin{eqnarray*}
\Om{(\Assume{a})} \together \Om{(\Assume{b})} &=& \Om{(\Assume{(a \join b)})}
\end{eqnarray*}
\end{lemma}
Appendix~\ref{S-proofs} provides the proof on page \pageref{P-assume-iter-conj-assume-iter}.

\subsection{Parallel}

Like weak conjunction, the parallel operator can be thought of as synchronising the atomic steps of two commands until either terminates, fails, or becomes infeasible, e.g.
\begin{displaymath}
\begin{array}{lll}
(a_1\Seq a_2 \Seq a_3 \Seq a_4) \parallel (b_1 \Seq b_2 \Seq b_3 \Seq b_4)
& = &
(a_1 \parallel b_1) \Seq (a_2 \parallel b_2) \Seq (a_3 \parallel b_3) \Seq (a_4 \parallel b_4) \\
(a_1\Seq a_2 \Seq a_3 \Seq a_4) \parallel (b_1 \Seq b_2 \Seq \bot)
& = &
(a_1 \parallel b_1) \Seq (a_2 \parallel b_2) \Seq \bot \\
(a_1\Seq a_2 \Seq a_3 \Seq a_4) \parallel (b_1 \Seq b_2 \Seq \top)
& = &
(a_1 \parallel b_1) \Seq (a_2 \parallel b_2) \Seq \top
\end{array}
\end{displaymath}
The definition of $a \parallel b$ for atomic steps $a$ and $b$ is not further axiomatised here to retain its generality,
but we show how it may be interpreted in different formalisms in Sections~\ref{S-rg-logic} and \ref{S-process-algebras}. 

For the parallel operator, the element $\atomid$, the atomic step that synchronises under parallel with any other atomic step $a\in \AtomicSteps$,
\begin{eqnarray}
\label{eqn-a-atomid}
  a \parallel \atomid& = & a ~,
\end{eqnarray}
can be interpreted as a placeholder for any one step taken by the environment. For example, prefixing a command $c$ with $\atomid$, i.e.\ $\atomid \Seq c$, defines a process that waits for one step, allowing any single environment step to take place, before behaving as $c$. The command $\Skip = \Om{\atomid}$ represents any finite or infinite number of steps taken by the environment and is the identity of parallel for arbitrary command $c\in\Commands$:
\begin{eqnarray*}
  c \parallel \Skip & = & c ~.
\end{eqnarray*}
This allows us to define, for any atomic step $a\in \AtomicSteps$, a program
$\Om{\atomid} \Seq a \Seq \Om{\atomid}$, 
that takes step $a$, but permits any possible environment behaviour both beforehand and afterwards.
(Note that the environment may interrupt the process forever with an infinite sequence of environment steps, thus preventing the process from taking step $a$.) 
Programs of this form are the building blocks of (interleaving) event-based languages, 
and are used in Section~\ref{S-process-algebras} to describe the interpretation of the parallel operator in process algebras.

The introduction of the atomic step $\atomid$, also provides us with an opportunity to further decompose the set of atomic steps, $\AtomicSteps$, into sub-algebras of program and environment steps, which we do in the following section.

\section{Program and environment steps}\label{S-pi-env}

\begin{figure}
 \centering
   \includegraphics[viewport=29 576 340 822,width=0.42\linewidth]{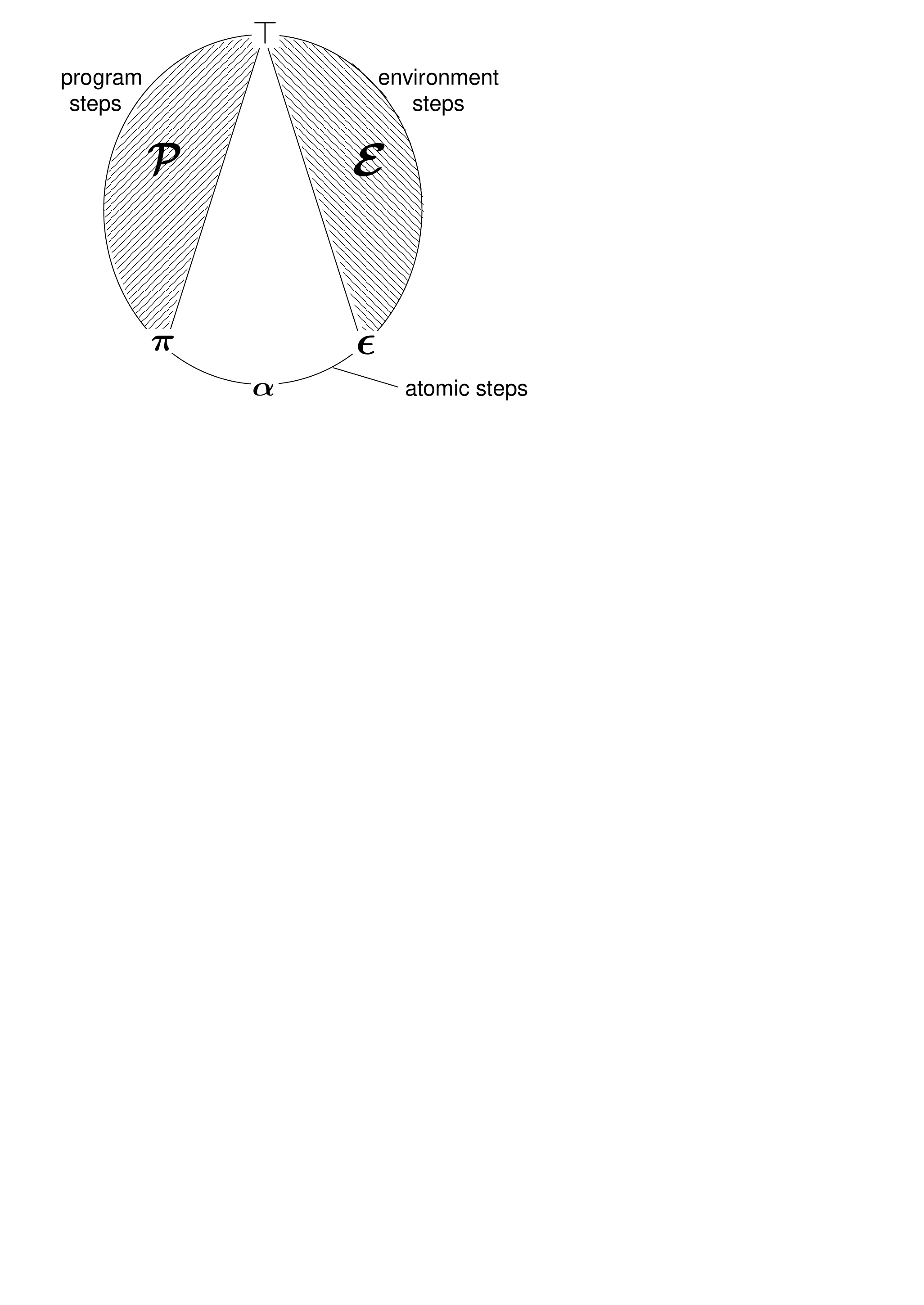}
  \caption{$\PCommands$ and $\ECommands$ as Boolean sub-algebras of $\AtomicSteps$}
  \label{fig:pienvlattice}
\end{figure}

Given that $\atomid$ represents any possible step taken by the environment, we define $\ECommands$ to be the subset of atomic steps $\e_i\in \AtomicSteps$ such that $\atomid \refsto \e_i \refsto \top$. These correspond the set of all possible \emph{environment steps}, and they can be shown to form a Boolean sub-algebra of atomic steps with disjunction $\nondet$, conjunction $\sqcup$, and with the negation of an environment step $\e_i$ within $\ECommands$ defined to be $(\anegate \e_i) \sqcup \atomid$.
We then define the complement of $\atomid$ within $\AtomicSteps$,
\begin{eqnarray}
\pibot =  \anegate \atomid \label{A-pibot-ebot}
\end{eqnarray}
to be the atomic step that can take any possible program step, and $\PCommands$ to be the subset of atomic steps $\pi_i \in \AtomicSteps$ such that $\pibot \refsto \pi_i \refsto \top$. These are then the set of all possible \emph{program steps}, and they can also be shown to form a Boolean sub-algebra of atomic steps.

Figure~\ref{fig:pienvlattice} illustrates the relation between program, environment and atomic steps.
Both sub-lattices have a bottom element, $\pibot$ and $\ebot$ respectively, and share the same top element, $\top$, with $\AtomicSteps$ and the command lattice $\Commands$.
Because both program steps and environment steps form a sub-lattice of the atomic steps $\AtomicSteps$, we have that the bottom element of $\AtomicSteps$, $\cstepd$, is refined by any program step as well as any environment step.
\vspace*{-2ex}
\\
\begin{minipage}[t]{0.5\textwidth}
\begin{eqnarray}
 \cstepd &~\refsto~& \pibot
\end{eqnarray}
\end{minipage}%
\begin{minipage}[t]{0.5\textwidth}
\begin{eqnarray}
 \cstepd &~\refsto~& \ebot
\end{eqnarray}
\end{minipage}
\vspace*{1ex}

As depicted in Figure~\ref{fig:pienvlattice}, 
program and environment steps have only the top element $\top$ in common (\ref{join-pibot-ebot}). 
All other atomic steps are either program or environment steps, 
or some non-deterministic choice over program and environment steps (\ref{choice-pi-env}). 
By definition (\ref{A-pibot-ebot}), the non-deterministic choice between the least program and environment step 
$\atomid \nondet \pibot = \atomid \nondet \anegate \atomid$ is $\cstepd$.
An atomic step negation applied to a arbitrary environment step, $\anegate \e_i$, 
results in any atomic step other than one $\e_i$ can perform 
and hence the non-deterministic choice of $\e_i$ and its negation equals the bottom of $\AtomicSteps$
(\ref{negate_e_inf_e}), 
which includes any possible program steps. 
Atomic negation applied to the bottom of $\PCommands$, $\anegate \pibot$, 
results in the non-deterministic choice over any atomic steps that are not program steps, 
which is the same as $\ebot$ (\ref{anegate-pibot}).  
\vspace*{-2ex}
\\
\begin{minipage}[t]{0.5\textwidth}
\begin{eqnarray}
 \pi_i ~\join~ \e_i &=& \top \label{join-pibot-ebot}\\
 \pi_i ~\nondet~ \e_i &\in& \AtomicSteps \label{choice-pi-env}
\end{eqnarray}
\end{minipage}%
\begin{minipage}[t]{0.5\textwidth}
\begin{eqnarray}
 \anegate \e_i ~\nondet~ \e_i  &=& \cstepd \label{negate_e_inf_e}\\
 \anegate \pibot &=& \ebot \label{anegate-pibot}
\end{eqnarray}
\end{minipage}
\vspace*{1ex}

The distinction between program and environment steps enables us to specify and prove some basic algebraic laws about guarantees and relies in Section~\ref{S-rely-guarantee}.

\subsection{Embedding in Isabelle}

As for the encoding of the Boolean algebra of tests, described in Section~\ref{S-Boolean-encoding}, program and environment steps are embedded into the lattice of commands via mappings, injective homomorphisms, from some Boolean algebras $\BooleanAlg_1$
and
$\BooleanAlg_2$ into the Boolean sub-algebras of environment and program steps, respectively:
$ \cestepd : \BooleanAlg_1 \rightarrow \AtomicSteps$
and 
$ \cpstepd : \BooleanAlg_2 \rightarrow \AtomicSteps$,
where usually $\BooleanAlg_1 = \BooleanAlg_2$
so that each program step has a matching environment step.

 With this mapping in place, the non-deterministic choice of two
 program steps as well as the join and the weak conjunction of two
 program steps can be computed on the level of the Boolean algebra,
 (\ref{pi-inf-pi}--\ref{pi-sup-pi}).
 The same set of lemmas holds for environment steps, (\ref{e-inf-e}--\ref{e-sup-e}).
 \\
 \begin{minipage}{0.5\textwidth}
 \begin{eqnarray}
   \cpstep{p \union q} &=& \cpstep{p} \nondet \cpstep{q}      \label{pi-inf-pi}\\
   \cpstep{p \inter q}   &=& \cpstep{p} \sqcup \cpstep{q}      \label{pi-sup-pi} \\
   \cpstep{\bot} &=& \top \label{pstep-bot} \\
   \cpstep{\top} &=& \cpstepd \label{pstep-top} \\
   p \subseteq q &\iff& \cpstep{q} \refsto \cpstep{p} \label{pstep-iso}
 \end{eqnarray}
 \end{minipage}%
 \begin{minipage}{0.5\textwidth}
 \begin{eqnarray}
   \cestep{p \union q} &=& \cestep{p} \nondet \cestep{q}      \label{e-inf-e}\\
   \cestep{p \inter q}   &=& \cestep{p} \sqcup \cestep{q}      \label{e-sup-e} \\
   \cestep{\bot} &=& \top \label{estep-bot} \\
   \cestep{\top} &=& \cestepd \label{estep-top} \\
   p \subseteq q &\iff& \cestep{q} \refsto \cestep{p} \label{estep-iso}
 \end{eqnarray}
 \end{minipage}

\section{Abstract specifications, guarantees and relies}\label{S-rely-guarantee} \label{SS-rely-guar}

In the rely/guarantee approach of Jones~\cite{CoJo07,Jones81d,Jones83a,Jones83b}, concurrent program specifications are traditionally formulated in terms of a  quintuple 
\begin{equation}\label{Hoare-triple}
\quintprgqc 
\end{equation}
which extends a Hoare triple with the rely $r$ and guarantee $g$ to handle concurrency \cite{Jones81d,Jones83a,Jones83b}.
The quintuple states that every program step of $c$ satisfies relation $g$ (on the program state) and that it terminates and establishes the initial-final-state relation $q$, provided it is executed from an initial state satisfying predicate $p$ and interference from the environment is bounded by relation $r$.

The synchronous refinement algebra can be used to abstractly represent specifications of this kind.
Predicates and the initial-final-state relation are abstracted using tests (Section~\ref{S-tests}), and relations $r$ and $g$ are treated as elements of \emph{some} Boolean algebras $\BooleanAlg_1$ and $\BooleanAlg_2$, respectively, so that $\cestep{r}$ represents the environment step satisfying $r$ and $\cpstep{g}$ is the program step satisfying $g$ (Section~\ref{S-pi-env}).

Instead of defining such a specification using a monolithic approach as in (\ref{Hoare-triple}), 
we decompose the two commitments (described in the above quintuple by $g$ and $q$) 
and the two assumptions ($p$ and $r$), into four separate constructs.
Initial-state assumptions ($p$) can already be represented using assertions (Section~\ref{S-assertions}), and so in this section we introduce strong specifications (that capture the end-to-end commitment $q$), program-step guarantees (describing guarantee $g$) and environment-step assumptions (describing rely $r$).
This simplifies reasoning by allowing the constructs to be treated separately, 
(e.g. strengthening a guarantee $g$ does not involve $p$, $r$ and $q$), 
as well as in combination.
They can be composed to give an overall program specification.

\subsection{The strong specification command}\label{S-specification}

A \emph{terminating command}, is one that only performs a finite number of program steps, but puts no constraints on the steps taken by its environment -- and hence could be interrupted forever by its environment. The most general non-terminating command, $\Term$, is defined by
\begin{eqnarray}\label{def-term}
\Term  &\sdef & \Fin{\cstepd} \Seq \Om{\atomid}~,
\end{eqnarray}
and, formally, we say that a command $c$ is terminating when $\Term \refsto c$.
Given a pair of tests $t_1, t_2\in \Tests$, a command of the form
\begin{eqnarray*}
t_1 \Seq \Term \Seq t_2
\end{eqnarray*}
represents a terminating command that when started in an initial state satisfying $t_1$, terminates (unless it is interrupted indefinitely by its environment) in a final state satisfying $t_2$. It is a simple example of an \emph{end-to-end specification}: a terminating command that constrains the initial and final state, but does not otherwise restrict the program or environment steps in any way. 
We call such specifications \emph{strong} because they must achieve their end-to-end constraints regardless of the behaviour of the environment, and they are usually not implementable unless they are weakened by making assumptions on the environment steps of the command. Environment assumptions may be combined with strong specifications using the weak-conjunction operator to produce feasible specifications. 
More formally, for $q \in \pset(\Tests \times \Tests)$, we refer to any command
\begin{eqnarray*}
\Spec{}{}{q} = \Nondet_{ (t_1,t_2)\in q} t_1 \Seq \Term \Seq t_2 
\end{eqnarray*}
as a \emph{strong specification command}, an abstraction of Morgan's specification statement~\cite{TSS}. A strong specification $\Spec{}{}{q}$ can be weakened to include an assumption that the initial state of the command satisfies test $p$, by pre-composing it with the assertion $\Assert{p}$, giving $\Assert{p} \Seq \Spec{}{}{q}$.

We can show that any strong specification is unaffected by weak conjunction or parallel composition with $\Term$, a property which is useful when decomposing program specifications into parallel implementations (see Section~\ref{S-parallel-intro}). 

\begin{lemma}[specification-terminates]\label{L-specification-terminates}
For arbitrary strong specification command $\Spec{}{}{q}$, we have that 
\begin{eqnarray*}
\Spec{}{}{q} \together \Term = \Spec{}{}{q}  & \textrm{~~and~~} & \Spec{}{}{q} \parallel \Term = \Spec{}{}{q} ~.
\end{eqnarray*}
\end{lemma}

\begin{proof}
Since synchronisation operators distribute over non-deterministic choices (\ref{A-sync-Inf-distrib}), the proof reduces to showing that for arbitrary tests $t_1,t_2 \in \Tests$ 
\begin{eqnarray}
(t_1 \Seq \Term \Seq t_2) \together \Term = t_1 \Seq \Term \Seq t_2
& \textrm {~~and~~} & 
(t_1 \Seq \Term \Seq t_2) \parallel \Term = t_1 \Seq \Term \Seq t_2 ~.
\label{L-specIdem}
\end{eqnarray}
This property trivially holds for weak conjunction ($\together$) since it behaves like conjunction for commands that refine $\Chaos$ (\ref{A-weakconj-nonfail}), and $\Chaos = \Om{\cstepd} \refsto \Term \refsto t_1\Seq \Term \Seq t_2$. 
For parallel composition we first reason that
\begin{eqnarray}
\Term \parallel \Term & = & \Term
\end{eqnarray}
because
\begin{displaymath}
\Term \parallel \Term
\Equals
(\Fin{\cstepd} \Seq \Om{\atomid}) \parallel (\Fin{\cstepd} \Seq \Om{\atomid})
\Equals*[using Lemma \ref{L-atomic-iteration-finite} and idempotence of choice ($\nondet$)]
\Fin{(\cstepd \parallel \cstepd)}\Seq (\Om{\atomid} \parallel \Fin{\cstepd} \Seq \Om{\atomid})
\Equals*[we have that $\cstepd \parallel \cstepd = \cstepd$ from (\ref{A-sync-cstepd}), and $\Skip = \Om{\atomid}$ is the identity of parallel
]
\Fin{\cstepd} \Seq \Fin{\cstepd} \Seq \Om{\atomid}
\Equals*[using property $\Fin{c}\Seq\Fin{c} = \Fin{c}$ (\ref{L-finite-twice})]
\Fin{\cstepd} \Seq \Om{\atomid}
\Equals
\Term
\end{displaymath}
and then that
\begin{eqnarray}
(t_1 \Seq \Skip \Seq t_2) \parallel\Term  & = & t_1 \Seq \Term \Seq t_2
\end{eqnarray}
which is shown by
\begin{displaymath}
(t_1 \Seq \Skip \Seq t_2) \parallel \Term 
\Equals
(t_1 \Seq \Om{\cestepd} \Seq t_2) \parallel (\Fin{\cstepd} \Seq \Om{\atomid})
\Equals*[distribute $t_1$ over parallel using Lemma \ref{L-test-command-sync-command}]
t_1 \Seq ((\Om{\cestepd} \Seq t_2) \parallel (\Fin{\cstepd} \Seq \Om{\atomid}))
\Equals*[using Lemma \ref{L-atomic-iteration-finite-omega}]
t_1 \Seq \Fin{(\cestepd \parallel \cstepd)} \Seq
(
  (t_2 \parallel \Fin{\cstepd} \Seq \Om{\atomid})
  \nondet
  (\Om{\cestepd} \Seq t_2 \parallel \Om{\atomid})
)
\Equals*[atomic step $\atomid$ and command $\Skip = \Om{\atomid}$ are the atomic-step identity and identity of parallel, respectively]
t_1 \Seq \Fin{\cstepd} \Seq
(
  (t_2 \parallel \Fin{\cstepd} \Seq \Om{\atomid})
  \nondet
  \Om{\cestepd} \Seq t_2
)
\Equals*[distribute $t_2$ over parallel using Lemma \ref{L-test-command-sync-command}]
t_1 \Seq \Fin{\cstepd} \Seq
(
  t_2 \Seq (\Nil \parallel \Fin{\cstepd} \Seq \Om{\atomid})
  \nondet
  \Om{\cestepd} \Seq t_2
)
\Equals*[we have $(\Nil \parallel (\Fin{\cstepd} \Seq \Om{\atomid})) = \Nil$ by unfolding iterations and applying synchronisation axioms (\ref{A-nil-sync-absorb}) and (\ref{A-nil-sync-atomic})]  
t_1 \Seq \Fin{\cstepd} \Seq
(
   t_2 \nondet \Om{\cestepd} \Seq t_2
)
\Equals*[by iteration unfolding, $\Om{\cestepd} \Seq t_2 \refsto t_2$]
t_1 \Seq \Fin{\cstepd} \Seq \Om{\cestepd} \Seq t_2
\end{displaymath}
which finally gives us
\begin{displaymath}
(t_1 \Seq \Term \Seq t_2) \parallel \Term
=
(t_1 \Seq \Skip \Seq t_2) \parallel \Term \parallel \Term 
=
(t_1 \Seq \Skip \Seq t_2) \parallel \Term 
=
(t_1 \Seq \Term \Seq t_2)
\end{displaymath}
which concludes our proof.
\end{proof}

\subsection{The guarantee command}\label{S-guarantee}

For a process to guarantee $g$, every atomic program ($\pstepd$) step made by the program must satisfy $g$.
A guarantee puts no constraints on the environment of the process.
A guarantee command, $\Guar{g}$, is defined in terms of the iteration of
a single step guarantee, $\pguard{g}$, defined as follows.
\vspace{-1ex}
\\
\begin{minipage}{0.5\textwidth}
\begin{eqnarray}
  \pguard{g} & \sdefs &\cpstep{g} \nondet \atomid 
\end{eqnarray}
\end{minipage}%
\begin{minipage}{0.5\textwidth}
\begin{eqnarray}
  \Guar{g}    & \sdefs & \Om{\pguard{g}} \label{eqn-def-guarantee}
\end{eqnarray}
\end{minipage}
\vspace*{1ex}
\\
A command $c$ with a guarantee of $g$ is represented by $(\Guar{g}) \together c$.
If $c$ only fails when its environmental assumptions are violated, $(\Guar{g}) \together c$ constrains the program steps of $c$ to satisfy $g$ until the assumptions of $c$ are broken.
For example, $(\Guar{g}) \together (\Assert{p}\Seq \Spec{}{}{q})$, is used to define a program that satisfies guarantee $g$ and strong specification $\Spec{}{}{q}$ if the initial state satisfies test $p$. No constraints are placed on its behaviour otherwise. 

In the theory of Jones, a guarantee on a process may be strengthened.
\begin{lemma}[strengthen-guarantee]\label{L-strengthen-guarantee}
If $g_1 \subseteq g_2$, then $\pguard{g_2} \refsto \pguard{g_1}$ and  $\Guar{g_2} \refsto \Guar{g_1}$.
\end{lemma}

\begin{proof}
If $g_1 \subseteq g_2$, then $\cpstep{g_2} \refsto \cpstep{g_1}$ and hence $\pguard{g_2} \refsto \pguard{g_1}$. Refinement of the guarantees follows by monotonicity of iteration.  
\end{proof}

A process that must satisfy both guarantee $g_1$ and guarantee $g_2$, must satisfy guarantee $g_1 \cap g_2$.

\begin{lemma}[combine-guarantees]\label{L-combine-guarantees}
For $g_1$, $g_2$ we have that
\begin{eqnarray} 
\pguard{g_1} \together \pguard{g_2} & = & \pguard{(g_1 \cap g_2)} \label{L-combine-restrict} \\
\Guar{g_1} \together \Guar{g_2} & = & \Guar{(g_1 \cap g_2)}     \label{L-combine-guarantee}
\end{eqnarray}
\end{lemma}

\begin{proof}
For single-step guarantees (\ref{L-combine-restrict}) we have 
\begin{displaymath}
\pguard{g_1} \together \pguard{g_2}
\Equals*[weak conjunction is conjunction for atomic steps (\ref{A-together-atomic}), expand single-step guarantees]
(\cpstep{g_1} \nondet \atomid) \sqcup (\cpstep{g_2} \nondet \atomid)
\Equals*[distribute conjunction over choices]
(\cpstep{g_1} \sqcup \cpstep{g_2}) \nondet
(\cpstep{g_1} \sqcup \atomid) \nondet
(\atomid \sqcup \cpstep{g_2}) \nondet
(\atomid \sqcup \atomid)
\Equals*[program and environment steps are disjoint (\ref{join-pibot-ebot}), homomorphism $\cpstepd$ ]
\cpstep{g_1 \cap g_2} \nondet \atomid
\Equals
\pguard{(g_1 \cap g_2)}
\end{displaymath}  
The proof of (\ref{L-combine-guarantee}) then follows from (\ref{L-combine-restrict}) using Corollary~\ref{L-atomic-either-sync} for synchronisation operator weak conjunction ($\together$).
\end{proof}

If we assume that commands only consist of atomic steps (which are
closed under the operators), then we have that an iterated atomic step
distributes over a sequence of commands,
$a^\omega \sync (c\Seq d) = (a^\omega \sync c) \Seq (a^\omega \sync d)$.
With this assumption, it follows that 
 guarantees also distribute over a sequence of commands.
\begin{eqnarray}
\label{eqn-guar-dist}
  (\Guar{g}) \together (c \Seq d) & = & ((\Guar{g}) \together c) \Seq ((\Guar{g}) \together d)
\end{eqnarray}

\subsection{The rely command}\label{S-rely}

A rely condition $r$ represents an assumption about environment steps.
After an environment step that does not satisfy $r$, 
i.e.\ a step that refines $\cestep{\prednegate{r}}$ where $\prednegate{r}$ is the complement of $r$,
the process may do anything,
which can be represented by it aborting.
Any other step, $\cnegate\cestep{\prednegate{r}} = \cpstepd \nondet \cestep{r}$, is allowed.
The rely command is defined in terms of a single step assumption,
itself defined in terms of the abstract command $\Assume{}$ (\ref{def-assume}) as follows.
\begin{eqnarray}
   \eassume{r} & \sdefs & \Assume{(\cpstepd \nondet \cestep{r})} 
               ~~~~~~~~ = \cpstepd \nondet \cestep{r} \nondet \cestep{\prednegate{r}} \Seq \bot \\
   \Rely{r} & \sdefs & \Om{\eassume{r}}  \label{def-rely}
\end{eqnarray}
A command $c$ with a rely $r$ imposed on its environment steps is given by $(\Rely{r}) ~\together~ c$. For example, $(\Rely{r})\together (\Guar{g})$, guarantees to only take program steps satisfying $g$ until the environment takes a step in which $r$ is violated. Similarly, $(\Rely{r})\together \Spec{}{}{q}$ terminates in a state satisfying $q$ when it is executed in an environment that only takes steps satisfying $r$. Alternatively, if it is executed in an environment that does take a step that violates $r$, then its behaviour from that point onwards is not constrained -- it may not terminate and has no obligation to satisfy $q$ if it does.

Weakening a rely condition allows more environment interference and hence is a refinement.
\begin{lemma}[weaken-rely]\label{L-weaken-rely}
If $r_0 \subseteq r_1$, then
$\Rely{r_0} \refsto \Rely{r_1}$.
\end{lemma}

\begin{proof}
This follows from anti-monotonicity of assumptions (Lemma~\ref{L-weaken-assume}) because 
\[
  \Assume{(\cpstepd \nondet \cestep{r_0})} \refsto \Assume{(\cpstepd \nondet \cestep{r_1})}
  ~~\Leftarrow~~
  \cpstepd \nondet \cestep{r_1} \refsto \cpstepd \nondet \cestep{r_0}
  ~~\iff~~
  r_0 \subseteq r_1
\]
and monotonicity of iteration.
\end{proof}

Weak conjunctions of rely conditions simplify in the following way.

\begin{lemma}[combine-relies]\label{L-combine-relies}
$(\Rely{r_1}) \together (\Rely{r_2}) = (\Rely{r_1 \inter r_2})$
\end{lemma}

\begin{proof}
Using Lemma (\ref{L-assume-iter-conj-assume-iter}) we have:
\begin{displaymath}
(\Rely{r_1}) \together (\Rely{r_2})
\Equals*[expand rely definition (\ref{def-rely})]
\Om{(\Assume{(\cpstepd \nondet \cestep{r_1})})}
~\together~
\Om{(\Assume{(\cpstepd \nondet \cestep{r_2})})}
\Equals*[simplify weak conjunction of assumptions using Lemma \ref{L-assume-iter-conj-assume-iter}]
\Om{(\Assume{
    (\cpstepd \nondet \cestep{r_1})
    \sqcup
    (\cpstepd \nondet \cestep{r_2})
})}
\Equals*[distribute conjunction over choices]
\Om{(\Assume{
    (
    (\cpstepd \sqcup \cpstepd) \nondet
    (\cpstepd \sqcup \cestep{r_2}) \nondet
    (\cestep{r_1} \sqcup \cpstepd) \nondet
    (\cestep{r_1} \sqcup \cestep{r_2})
    )
})}
\Equals*[program and environment steps are disjoint (\ref{join-pibot-ebot}), homomorphism $\cestepd$]
\Om{(\Assume{
    (\cpstepd \nondet \cestep{r_1 \cap r_2})
  })}
\Equals*[fold rely definition (\ref{def-rely})]
(\Rely{r_1 \cap r_2})
\end{displaymath}
\end{proof}

\subsection{Abstract rely/guarantee quintuple}\label{S-Hoare-logic}

Combining initial-state assumptions, strong specifications, guarantee and rely constructs together we get that
\begin{displaymath}
\Assert{p}\Seq ((\Rely{r})\together (\Guar{g}) \together \Spec{}{}{q})
\end{displaymath}
is a specification command that satisfies guarantee $g$ and strong specification $\Spec{}{}{q}$ until the environment takes a step that violates $r$, given that the command is executed from a state satisfying test $p$. Any command $c$ is said to satisfy this specification if it is a refinement of it.
This enables algebraic reasoning about rely/guarantee quintuples.

\section{Interpretation for shared-memory concurrency}\label{S-rg-logic}\label{S-intro-par}

For shared-memory concurrency on a state-space $\Sigma$, tests are instantaneous commands, taking no atomic steps, that either have no effect on the execution of the program (behaving like $\Nil$) if a given predicate $p\in \pset(\Sigma)$ is satisfied, and are infeasible (behaving like $\top$) otherwise.
For this model, we define the injective homomorphism $\cgdd$, that introduces the Boolean sub-algebra of tests, to be a mapping from the Boolean algebra of predicates, $(\pset(\Sigma), \union, \inter, \bar{~~}, \bot, \top)$, to commands:
\begin{equation*}
\cgdd : \pset(\Sigma) \fun \Commands ~.
\end{equation*}
 
A program or environment step is an atomic transition of the state of the system from a before state, $\sigma$, to an after state, $\sigma'$, satisfying some relation $r\in \pset(\Sigma{\times}\Sigma)$, that is made by either the program or its environment, respectively. For before-states outside of the domain of $r$, the atomic step is infeasible, behaving like program $\top$.
These atomic step commands can be thought of as lifting the primitive program and environment steps on pairs of states, e.g. $\pstep{\sigma, \sigma'}$ and $\estep{\sigma, \sigma'}$, used in Aczel's original trace semantics, to the level of commands. 

We therefore instantiate injective homomorphisms $\cestepd$ and $\cpstepd$, that introduce the Boolean sub-algebras of environment and program steps, to both be mappings from the Boolean algebra of relations, $(\pset(\Sigma{\times}\Sigma), \union, \inter, \bar{~~}, \bot, \top)$, to commands: 
\begin{equation*}
\cestepd  : \pset(\Sigma{\times}\Sigma) \fun \AtomicSteps \textrm{~~and~~}
\cpstepd : \pset(\Sigma{\times}\Sigma) \fun \AtomicSteps 
\end{equation*}

Under sequential composition, program and environment steps are also specified to follow the additional axioms in which $p \dres r$ stands for the domain restriction of the relation $r$ to the set $p$ and $r \rres p$ stands for the range restriction of $r$ to $p$.\\
\begin{minipage}{0.5\textwidth}
\begin{eqnarray}
  \cgd{p} \Seq \cpstep{r} & = & \cpstep{p \dres r} \\
  \cpstep{r \rres p}  \Seq \cgd{p} & = & \cpstep{r \rres p}
\end{eqnarray}
\end{minipage}%
\begin{minipage}{0.5\textwidth}
\begin{eqnarray}
  \cgd{p} \Seq \cestep{r} & = & \cestep{p \dres r} \\
  \cestep{r \rres p}  \Seq \cgd{p} & = & \cestep{r \rres p}
\end{eqnarray}
\end{minipage}\vspace{1ex}\\
They encode the fact that either a program or environment step, $\cpstep{r}$ or $\cestep{r}$, is infeasible for states outside of the domain of its relation $r$, and establishes a state satisfying the range restriction of $r$. 
The above axioms ensure the final state of one step matches the initial state of the next step.

\subsection{Interleaving parallel}

We define parallel composition as the interleaving of steps, in which environment steps and program steps synchronise using the following axioms, that follow Aczel's original trace-based formalisation:
\begin{eqnarray}
	\cpstep{r_1} \pl \cestep{r_2} 
	&=& 
	\cpstep{r_1 \inter r_2} 
	\label{eqn-aczel-pe}
	\\
	\cestep{r_1} \pl \cestep{r_2} 
	&=& 
	\cestep{r_1 \inter r_2} 
	\label{eqn-aczel-ee}
	\\
	\cpstep{r_1} \pl \cpstep{r_2} 
	&=& 
	\top
	\label{eqn-aczel-pp}
\end{eqnarray}
The first axiom matches a program step with an environment step, possibly narrowing the relation to conform to both relations.  If no step satisfies $r_1$ and $r_2$ then the parallel composition gives $\top$.
The second axiom matches environment steps of each process to become an environment step of their composition.
The final axiom prevents synchronisation on program steps.
Using these synchronisation axioms for parallel we have, for example, that for relations $g_1, g_2, r_1$ and $r_2$
\begin{eqnarray*}
(\cpstep{g_1} \nondet \cestep{r_1}) \parallel (\cpstep{g_2} \nondet \cestep{r_2})
& ~~=~~ &
\cpstep{(g_1 \cap r_2) \cup (r_1 \cap g_2)} \nondet
\cestep{r_1 \cap r_2}
\end{eqnarray*}
We can also evaluate the parallel composition of relies and guarantees. The following lemma that shows that a rely of $r$ allows parallel behaviour for which every step guarantees $r$.

\begin{lemma}[rely-guar]\label{L-rely-guar} 
\[
  \Rely{r} = (\Rely{r}) \parallel (\Guar{r})
\]
\end{lemma}
\begin{proof}
The detailed proof expands the definitions of $\Rely{r}$ and $\Guar{r}$ and simplifies.
\begin{displaymath}
(\Rely{r}) \parallel (\Guar{r})
\Equals*[expand rely and guarantee definitions using (\ref{def-rely}) and (\ref{eqn-def-guarantee})]
\Om{(\cpstepd \nondet \cestep{r} \nondet \cestep{\prednegate{r}} \Seq \bot)} \parallel 
\Om{(\cpstep{r} \nondet \cestepd)}
\Equals*[by Lemma~\ref{L-iterations-with-abort}]
\Om{((\cpstepd \nondet \cestep{r}) \parallel (\cpstep{r} \nondet \cestepd))} \Seq 
(\Nil \nondet (\cestep{\prednegate{r}} \parallel (\cpstep{r} \nondet \cestepd)) \Seq \bot)
\Equals*[applying program and environment step synchronisation axioms (\ref{eqn-aczel-pe}) and (\ref{eqn-aczel-ee})]
\Om{(\cpstepd \nondet \cestep{r})} \Seq (\Nil \nondet \cestep{\prednegate{r}} \Seq \bot)
\Equals*[by Lemma~\ref{L-iterated-assumption} and folding the rely definition (\ref{def-rely})]
\Rely{r}
\end{displaymath}
\end{proof}

\subsection{Parallel introduction law}\label{S-parallel-intro}

The objective of this section is to give a refinement law for taking a command of the form $(\Rely{r}) \together (c \together d)$ and refining it into a parallel implementation in which the two parallel branches perform commands $c$ and $d$, respectively. 
This is achieved with the help of relies and guarantees.  
The proof requires some constraints on the commands $c$ and $d$ for which this is possible, namely they need to terminate (\ref{cterm}) and also be unaffected when put in parallel with an arbitrary terminating command (\ref{parIdem}):
\vspace{-2ex}\\
\begin{minipage}{0.5\textwidth}
\begin{eqnarray}
  c \together \Term = c \label{cterm}
\end{eqnarray}
\end{minipage}%
\begin{minipage}{0.5\textwidth}
\begin{eqnarray}
c \parallel \Term = c \label{parIdem}
\end{eqnarray}
\end{minipage}\vspace{1ex}
To verify the parallel introduction law below, we use the following lemma, that introduces a parallel guarantee that terminates.
The proof makes use of the interchange axiom between weak conjunction and parallel (\ref{conjunction-interchange-parallel}).
\begin{lemma}[parallel-guarantee]\label{L-parallel-guarantee}
For any command $c$ satisfying (\ref{parIdem}), 
\[
  (\Rely{r}) \together c \refsto 
    ((\Rely{r \union r1}) \together c) \parallel ((\Guar{r \union r1}) \together \Term)
\]
\end{lemma}

\pagebreak[2]

\begin{proof}
\[
  (\Rely{r}) \together c 
 \Refsto*[using Lemma~\ref{L-weaken-rely} to weaken the rely]
  (\Rely{r \union r1}) \together c
 \Equals*[by Lemma~\ref{L-rely-guar} and assumption (\ref{parIdem})]
  ((\Rely{r \union r1}) \parallel (\Guar{r \union r1})) \together (c \parallel \Term)
 \Refsto*[by (\ref{conjunction-interchange-parallel})] 
  ((\Rely{r \union r1}) \together c) \parallel ((\Guar{r \union r1}) \together \Term)
\]
\end{proof}

The parallel introduction law is an abstract version of that of Jones. 
The main difference is that here it is expressed based on our synchronous algebra primitives  and hence an algebraic proof is possible.
\begin{lemma}[introduce-parallel]\label{L-introduce-parallel}
For commands $c$ and $d$ both satisfying (\ref{cterm}) and (\ref{parIdem}), 
\[
  (\Rely{r}) \together (c \together d) \refsto 
    ((\Rely{r \union r_0}) \together (\Guar{r_1}) \together c) \parallel
    ((\Rely{r \union r_1}) \together (\Guar{r_0}) \together d)~.
\]
\end{lemma}

\begin{proof}
The proof first splits the specification into two conjoined specifications before
making use of Lemma~\ref{L-parallel-guarantee} to introduce parallel guarantee for each specification.
Once in this form the interchange axiom can be used to match up the relies and guarantees 
before a final simplification.
\[
  (\Rely{r}) \together (c \together d)
 \Equals*[as $\together$ is idempotent, associative and commutative]
  (\Rely{r}) \together c \together (\Rely{r}) \together d
 \Refsto*[by Lemmas~\ref{L-parallel-guarantee} and \ref{L-strengthen-guarantee}, twice]
  (((\Rely{r \union r_0}) \together c) \parallel ((\Guar{r_0}) \together \Term)) \together 
  (((\Guar{r_1}) \together \Term) \parallel ((\Rely{r \union r_1}) \together d))
 \Refsto*[by (\ref{conjunction-interchange-parallel})]
  ((\Rely{r \union r_0}) \together c \together (\Guar{r_1}) \together \Term) \parallel
  ((\Rely{r \union r_1}) \together d \together (\Guar{r_0}) \together \Term)
 \Equals*[by~(\ref{cterm})]
  ((\Rely{r \union r_0}) \together (\Guar{r_1}) \together c) \parallel
  ((\Rely{r \union r_1}) \together (\Guar{r_0}) \together d)
\]
\end{proof}

\subsection{Strong specifications}

Given a relation $q\in \pset(\Sigma{\times}\Sigma)$, 
we introduce the following short-hand (converting from a relation on states $q$ to a relation on tests) 
to refer to the strong specification statement (Section~\ref{S-specification}) 
that establishes $q$ between its initial and final states:
 \begin{eqnarray*}  
    \RSpec{}{}{q} & \sdefs & \Spec{}{}{ \{ \sigma \in \Sigma \spot (\cgd{\{\sigma\}},
                                              \cgd{\{\sigma' \in \Sigma \;\!|\!\; (\sigma,\sigma') \in q \} } ) \} } \\
    &=&\Nondet_{\sigma \in \Sigma} \cgd{\{\sigma\}} \Seq \Term \Seq 
              \cgd{\{\sigma' \in \Sigma \;\!|\!\; (\sigma,\sigma') \in q \}}
         \label{defn-spost}
 \end{eqnarray*}
Using this notation, we can show that following lemma holds.

\begin{lemma}[conjoin-postcondition]\label{L-conjoin-postcondition}
For relations $q_0, q_1 \in \pset(\Sigma{\times}\Sigma)$, 
\begin{math}
  \RSpec{}{}{q_0 \inter q_1} = \RSpec{}{}{q_0} \together \RSpec{}{}{q_1}
\end{math}
holds.
\end{lemma}
The detailed proof expands the definition of the specifications.
Informally $\RSpec{}{}{q_0}$ allows any terminating behaviour that end-to-end satisfies $q_1$
and $\RSpec{}{}{q_1}$ likewise but for $q_1$,
and hence their conjunction must agree on all behaviours and hence
satisfies $q_0 \inter q_1$ end-to-end.

We can instantiate Law \ref{L-introduce-parallel} to introduce the parallel operator when refining a specification, i.e., 
\begin{equation}\label{L-introduce-parallel-rspec}
  (\Rely{r}) \together \RSpec{}{}{q_0 \inter q_1} \refsto 
    ((\Rely{r \union r_0}) \together (\Guar{r_1}) \together \RSpec{}{}{q_0}) \parallel
    ((\Rely{r \union r_1}) \together (\Guar{r_0}) \together \RSpec{}{}{q_1})
\end{equation}
since  $\RSpec{}{}{q_0 \inter q_1} = \RSpec{}{}{q_o} \together \RSpec{}{}{q_1}$ (by Lemma~\ref{L-conjoin-postcondition}), and specification statements $\RSpec{}{}{q_o}$ and $\RSpec{}{}{q_1}$ satisfy the requirements on $c$ and $d$ (by Lemma \ref{L-specification-terminates}).

\newcommand{\synchpstepe}{\synchpstep{e}}
\renewcommand{\synchpstepE}{\synchpstep{E}}
\newcommand{\ateve}{\atev{e}}
\newcommand{\atevf}{\atev{f}}
\newcommand{\atevec}{\atev{\ecomple}}
\newcommand{\ecomple}{\ecompl{e}}
\newcommand{\spstepe}{\spstep{e}}
\newcommand{\spstepf}{\spstep{f}}
\newcommand{\spstepec}{\spstep{\ecomple}}
\newcommand{\plEE}{\plE{E}}
\newcommand{\plEe}{\plE{\{e\}}}

\newcommand{\complE}{\compl{E}}

\newcommand{\atomice}{\atomic{e}}

\newcommand{\refaxiom}[1]{(\ref{A-#1})}

\newcommand{\synced}[1]{\widehat{#1}}
\newcommand{\syncede}{\synced{e}}

\newcommand{\phipi}[1]{\widehat{\phi}_{#1}}
\newcommand{\phipiE}{\phipi{E}}
\newcommand{\phipie}{\phipi{\{e\}}}

\newcommand{\SRAResSet}[1]{\synced{#1} \union \compl{#1}}
\newcommand{\SRAResSete}{\{\synced{e}\} \union \compl{\{e\}}}

\newcommand{\SRARes}[2]{#2 \together (\Guar{#1})}
\newcommand{\SRAResE}[1]{\SRARes{\SRAResSet{E}}{#1}}
\newcommand{\SRARese}[1]{\SRARes{\SRAResSete}{#1}}

\section{Interpretation for event-based communication in process algebras}\label{S-process-algebras}

In the process algebra domain, processes communicate via a set of synchronisation events, $Event$, in contrast to 
processes in a shared memory concurrency model which
interleave operations on (sets of pairs of) states.
We may build a core process algebra from the basic operators, with the addition of
a set of atomic program steps $\synchpstepE$
that model a process engaging in one of the abstract events $e \in E$, for $E
\subseteq Event$. 
We assume the set $\Event$ includes at least the silent event $\silent$.
For ease of presentation we let $\spstepe$ abbreviate $\synchpstep{\{e\}}$, and
typically define our axioms on the singleton event-set case; the
definitions may be straightforwardly lifted to sets of events.

Rather than the interleaving interpretation of parallel in
\refsect{rg-logic}, in the process algebra domain program steps can be
combined.  We assume axioms \refeqn{aczel-pe} and \refeqn{aczel-ee}, but
replace \refeqn{aczel-pp} with language-specific definitions, as shown
below.  In a sense, the fundamental difference between the process algebras
CCS and CSP, and between process algebras and the shared-memory domain, is in
the interpretation of the parallel synchronisation of program steps
(of course, the domain type of the steps themselves is also fundamental).

We extend the core algebra to give two types of abstract interprocess communication: 
CCS-style binary synchronisation
and 
CSP-style multi-way synchronisation.
Note that to fully encode CSP in our algebra we would need to model \emph{external choice},
which behaves similarly to CCS's choice operator
(our non-deterministic choice operator corresponds to CSP's \emph{internal choice} operator).
Neither internal nor external choice satisfies conjunctivity of sequential composition (\ref{A-seq-distr-left}). 
The full encoding of these process algebras in our synchronous algebra is ongoing work.
See \cite{FM2016atomicSteps} for a discussion on the relationship between the process algebra SCCS and our algebra.

\subsection{Foundations}

To simplify the discussion we make the following definition.
\begin{equation}
	\ateve \sdef 
		\Om{\synchestepE} \Seq
		\spstepe \Seq
		\Om{\synchestepE}
	\label{eqn-defn-atomic}
\end{equation}
This models a process engaging in event $e$, and
is the building block of event-based languages:
we interpret both prefixing in CCS ($e.c$) and CSP ($e \fun c$) as $(\ateve \Seq c)$.
The event is preceded and succeeded by steps of the environment, similar to 
\emph{asynchronising} in Synchronous CCS \cite{Milner83} (discussed in \cite{CaC}), allowing the potential for interleaving.

The following lemma gives three possibilities for parallel actions: synchronisation or interleaving (the latter in one of two ways).
This is similar to a fundamental axiom of communication from ACP \cite{BergstraKlop84}, 
which defines parallel composition in terms of a
\emph{left-merge} operator.
\begin{lemma}[atomic-interleaving]\label{L-atomic-interleaving}
If $\spstep{e} \parallel \spstep{f} = \spstep{g}$, then
\(
  \atomic{e} \parallel \atomic{f} 
  ~=~ 
  \atomic{g} \nondet \atomic{e} \Seq \atomic{f} \nondet \atomic{f} \Seq \atomic{e}~.
\)
\end{lemma}

\newcommand{\omsid}{\Om{\atomid}}
\newcommand{\refcorollary}[1]{Corollary~\ref{C-#1}}

This is an instance of the following more general lemma.
\begin{lemma}[prefixed-interleaving]\label{L-prefixed-interleaving}
If $\spstep{e} \parallel \spstep{f} = \spstep{g}$,
and assuming that $c$ satisfies $\omsid \Seq c = c$ and $\omsid \Seq d = d$,
\begin{equation*}
  \atomic{e} \Seq c \parallel \atomic{f} \Seq d 
  ~=~ 
  \atomic{g} \Seq (c \pl d) \nondet 
  	\atomic{e} \Seq (c \pl \atomic{f} \Seq d) \nondet 
  	\atomic{f} \Seq (\atomic{e} \Seq c \pl d)~.
\end{equation*}
\end{lemma}

\begin{proof}
\[
  	\atomic{e} \Seq c \parallel \atomic{f} \Seq d
 	\Equals*[definition \refeqn{defn-atomic}]
  	(\omsid \Seq \spstepe \Seq \omsid \Seq c) 
  	\parallel 
  	(\omsid \Seq \spstepf \Seq \omsid \Seq d) 
 	\Equals*[\refcorollary{atomic-iteration-either}, and $\cestepd \pl \cestepd = \cestepd$]
  	\omsid \Seq (
		(
		  	(\spstepe \Seq \omsid \Seq c) 
		  	\parallel 
		  	(\spstepf \Seq \omsid \Seq d) 
		)
		\nondet
		(
		  	(\spstepe \Seq \omsid \Seq c) 
		  	\parallel 
		  	(\atomid \Seq \omsid \Seq \spstepf \Seq \omsid \Seq d) 
		)
		\nondet
		(
		  	(\atomid \Seq \omsid \Seq \spstepe \Seq \omsid \Seq c) 
		  	\parallel 
		  	(\spstepf \Seq \omsid \Seq d) 
		)
 	\Equals*[Synchronise initial steps, from assumption and \refeqn{a-atomid}]
  	\omsid \Seq (
		(
		  	\spstep{g} \Seq 
				(
				\omsid \Seq c
		  		\parallel 
		  		\omsid \Seq d
				)
		)
		\nondet
		(
		  	\spstepe \Seq 
				(
				(\omsid \Seq c)
		  		\parallel 
		  		(\omsid \Seq \spstepf \Seq \omsid \Seq d)
				)
		)
		\nondet
		(
		  	\spstepf \Seq 
				(
				(\omsid \Seq \spstepe \Seq \omsid \Seq c)
		  		\parallel 
		  		(\omsid \Seq d)
				)
		)
 	\Equals*[Simplify using definition \refeqn{defn-atomic} and the assumptions $\omsid \Seq c = c$ and $\omsid \Seq d = d$]
  	\omsid \Seq (
		  	\spstep{g} \Seq 
				(
				c
		  		\parallel 
		  		d
				)
		\nondet
		(
		  	\spstepe \Seq 
				(
				c
		  		\parallel 
		  		\atevf \Seq d
				)
		)
		\nondet
		(
		  	\spstepf \Seq 
				(
				\ateve \Seq c
		  		\parallel 
		  		d
				)
		)
 	\Equals*[if $\omsid \Seq c = c$ and $\omsid \Seq d = d$ then $c \pl d = \omsid \Seq (c \pl d)$ by \reflemma{atomic-iteration-either}, noting that $\ateve \Seq c = \omsid \Seq \ateve \Seq c$]
  	\omsid \Seq (
		  	\spstep{g} \Seq \omsid \Seq
				(
				c
		  		\parallel 
		  		d
				)
		\nondet
		(
		  	\spstepe \Seq \omsid \Seq
				(
				c
		  		\parallel 
		  		\atevf \Seq d
				)
		)
		\nondet
		(
		  	\spstepf \Seq \omsid \Seq
				(
				\ateve \Seq c
		  		\parallel 
		  		d
				)
		)
 	\Equals*[distibute initial $\omsid$; definition \refeqn{defn-atomic}]
		  	\atev{g} \Seq 
				(
				c
		  		\parallel 
		  		d
				)
		\nondet
		  	\ateve \Seq
				(
				c
		  		\parallel 
		  		\atevf \Seq d
				)
		\nondet
		  	\atevf \Seq
				(
				\ateve \Seq c
		  		\parallel 
		  		d
				)
\]
\end{proof}

\newcommand{\Rename}[2]{#2[#1]}
\newcommand{\Renamephi}[1]{\Rename{\phi}{#1}}
We introduce a \emph{renaming} operator, as in CCS.  
The command $\Rename{\phi}{c}$, where $\phi$ is a total function on atomic actions, is defined by the following axioms.
\begin{eqnarray}
	\Renamephi{\Nil}
	&=&
	\Nil
	\\
	\Renamephi{(c_1 \nondet c_2)}
	&=&
	\Renamephi{c_1} \nondet \Renamephi{c_2}
	\label{eqn-dist-rename-choice}
	\\
	\Renamephi{(c_1 \Seq c_2)}
	&=&
	\Renamephi{c_1} \Seq \Renamephi{c_2}
	\\
	\Renamephi{a}
	&=&
	\phi(a)
\end{eqnarray}

\subsection{Communication in CCS}
\label{S-ccs-encoding}

In CCS each non-silent event $e$ has a complementary event $\ecomple$.
(We use $\ecomple$ rather than Milner's $\overline{e}$ to avoid confusion with set complement.)
A program step $\spstepe$ and its corresponding complementary program step $\spstepec$
may synchronise to become a silent step.
\begin{eqnarray}
	\spstepe \pl \spstepec = \ssilent
	\label{eqn-ccs-pp}
\end{eqnarray}
All other combinations of program steps result in $\top$.

Using an instantiation of \reflemma{atomic-interleaving}, from
\refeqn{ccs-pp} we may
derive
\begin{eqnarray}
  \ateve \pl \atevec ~ = ~ \atevs \nondet \ateve \Seq \atevec \nondet \atevec \Seq \ateve~.
  \label{eqn-ccs-atev-sync}
\end{eqnarray}
As such, events may synchronise \emph{or} interleave.
In CCS the restriction operator 
$\Res{E}{c}$,
where $E$ is a set of $\Event$s,
may be employed to exclude the final two interleaving options and hence
force processes to synchronise and generate a silent step.
Restriction may be defined straightforwardly using 
weak conjunction ($\together$) to forbid events
in $E$, along with the concept of \emph{guarantees} from \refsect{pi-env}.
\begin{eqnarray}
	\Res{E}{c} &\sdef& c \together (\Guar{\compl{E}})
	\label{eqn-ccs-res}
\end{eqnarray}
This definition restricts $c$ to just behaviours outside of $E$.

The behaviour of an atomic event inside a restriction is given by the following
lemmas, which follow straightforwardly from \refeqn{def-guarantee} and
\refeqn{defn-atomic}.
\begin{lemma}[atomic-restriction]\label{L-atomic-restriction}
If $e \in E$ and $f \notin E$,
\begin{eqnarray*}
	\atomice \together \Guar{E} = \atomice
	\qquad
	\atomic{f} \together \Guar{E} = \omsid \Seq \top
\end{eqnarray*}
\end{lemma}
The command $\omsid \Seq \top$ allows the environment to make
some number of steps before becoming infeasible.  As such for atomic actions
it behaves similarly to $\top$ for general commands.
\begin{equation}
\label{eqn-omsidtop-elim}
	\atomice \nondet \omsid \Seq \top
	=
	\atomice
\end{equation}
We may now show that
a synchronisation within the corresponding restriction results in a silent step.
\begin{lemma}[CCS-synchronise]\label{L-ccs-synchronise}
\begin{equation*}
  	\Res{\{e, \ecomple\}}{ (\ateve \pl \atevec) } 
	= 
  	\atevs
\end{equation*}
\end{lemma}
\begin{proof}
\[
  	\Res{\{e, \ecomple\}}{
  		(\ateve \pl \atevec)
	} 
 	\Equals*[from \refeqn{ccs-atev-sync}]
  	\Res{\{e, \ecomple\}}{
  		(\atevs \nondet \ateve \Seq \atevec \nondet \atevec \Seq \ateve)
  	}
 	\Equals*[definition \refeqn{ccs-res}]
  	(\atevs \nondet \ateve \Seq \atevec \nondet \atevec \Seq \ateve)
	\together \Guar{\compl{\{e, \ecomple\}}}
 	\Equals*[Distribute weak conjunction over choice \refaxiom{sync-Inf-distrib}]
  	\SRARes{\compl{\{e, \ecomple\}}}{
  		\atevs 
  	}
	\nondet 
  	\SRARes{\compl{\{e, \ecomple\}}}{
		(\ateve \Seq \atevec) 
  	}
	\nondet 
  	\SRARes{\compl{\{e, \ecomple\}}}{
		(\atevec \Seq \ateve)
  	}
 	\Equals*[Distribute guarantee over sequential \refeqn{guar-dist}
	    $\cross 2$; \reflemma{atomic-restriction} $\cross 3$; simplify]
  	\atevs 
	\nondet 
  	\omsid \Seq \top
	\nondet 
  	\omsid \Seq \top
 	\Equals*[from \refeqn{omsidtop-elim}]
  	\atevs
\]
\end{proof}

\reflemma{prefixed-interleaving} and \reflemma{ccs-synchronise}
are the foundation for proving communication behaviours of CCS,
such as the following.
\[
	\Res{\{e, \ecomple\}}{
		(e.c \pl \ecomple.d)
	}
	~=~ 
	\iota.\Res{\{e,\ecomple\}}{(c \pl d)}
\]

\subsection{Communication in CSP}
\label{S-csp-encoding}

CSP-style multi-way communication allows 
any number of processes (not just two) to synchronise on an event.
The key axiom in the interpretation of parallel is, again, in the behaviour
of two program steps.%
\footnote{
The encoding given in \cite{FM2016atomicSteps} allowed synchronisation
between two processes on the alphabet $E$, but prevented a third process
(the environment) from engaging in $e \in E$, because environment steps
were removed from the traces.  
}
We define two program steps on event $e$ to merge into a single ``synchronised'' event, $\syncede$.
The set of events is extended with new events $\syncede$ for every event $e$, and
introduce a step $\sestep{e}$, where $\sestepE \refsto \sestep{e}$ for all $e \in \Event$.
\begin{equation}
	\spstepe \pl \spstepe ~=~ \spstep{\syncede}
	\quad \mbox{for $e \neq \silent$}
\label{eqn-csp-pp} 
\end{equation}
The synchronisation tag marks the event as having been the result of a
synchronisation, but because further synchronisation is possible, it is not
renamed to a silent step as in CCS.  The tag distinguishes
synchronised events from unsynchronised (interleaved) events.

Using an instantiation of \reflemma{atomic-interleaving}, from
\refeqn{csp-pp} we may
derive (cf. \refeqn{ccs-atev-sync})
\begin{eqnarray}
  \ateve \pl \ateve ~ = ~ \atev{\syncede} \nondet \ateve \Seq \ateve ~.
  \label{eqn-csp-atev-sync}
\end{eqnarray}

\newcommand{\cspResSet}[1]{\widetilde{#1}}
\newcommand{\cspResE}{\cspResSet{E}}
\newcommand{\cspRese}{\cspResSet{\{e\}}}
\renewcommand{\cspResSet}[1]{#1 \cup \synced{\compl{#1}}}

Now we define CSP's parallel operator parameterised with alphabet $E$ to allow
synchronisation on events in $E$ only, while only events not in $E$ may be
interleaved.  
The synchronisation tag is stripped using the renaming $\phipiE$.
\begin{eqnarray}
	c_1 \plEE c_2
	&\sdef&
	\Rename{\phipiE}{(\SRAResE{(c_1 \pl c_2)})}
	\label{eqn-csp-pl-defn}
	\\
	\mbox{where}
	\quad
	\phipiE(\spstep{\syncede}) &=& \spstepe 
			\qquad \mbox{if $e \in E$} 
	\label{eqn-defn-phipiE-in}
	\\
	\phipiE(a) &=& a 
			\qquad \quad~\mbox{for all other steps}
	\label{eqn-defn-phipiE-out}
\end{eqnarray}
Hence CSP synchronisation on alphabet $E$ is defined in terms of basic parallel, 
with a guarantee (similar to a CCS restriction) that only steps in $E$ may synchronise and only steps in the complement of $E$ may interleave, inside a renaming.
The notation $\synced{E}$ is the set of events
obtained by applying the synchronisation tag to elements in $E$, i.e., $\{\synced{e} | e \in E\}$.
Events in $E$ that have synchronised are renamed back to normal $\spstepe$ steps by the renaming $\phipiE$.  
Thus, the environment cannot determine whether an event is
the result of sychronisation or interleaving.%
\footnote{
For completeness, note that we could define CCS's parallel composition using a similar structure to that above for CSP, where instead
of \refeqn{ccs-pp} we define
$
	\spstepe \pl \spstepec = \spstep{\syncede}
$, and hence make the treatment closer to \refeqn{csp-pp}.  Such a definition may allow closer comparison of the two algebras; in this paper we instead
adapt at the more fundamental level of the definition of parallel composition.
}

Some of the basic communication properties from CSP follow from the above
definitions and the atomic algebra.
\begin{lemma}[atomic-sync]\label{L-atomic-sync}
\begin{equation*}
	\ateve \plE{\{e\}} \ateve
	=
	\ateve
\end{equation*}
\end{lemma}
\begin{proof}
\[
	\ateve \plE{\{e\}} \ateve
 	\Equals*[definition \refeqn{csp-pl-defn}]
	\Rename{\phipie}{(\SRARese{(\ateve \pl \ateve)})} 
	\Equals*[from \refeqn{csp-atev-sync}]
	\Rename{\phipie}{(\SRARese{(\atev{\syncede} \nondet \ateve \Seq \ateve)})}
 	\Equals*[distribute weak conjunction and renaming over choice \refaxiom{sync-Inf-distrib}; \refeqn{dist-rename-choice}]
	\Rename{\phipie}{(\SRARese{
		\atev{\syncede}
	})}
	\nondet 
	\Rename{\phipie}{(\SRARese{
		\ateve \Seq \ateve})
	}
	\Equals*[\reflemma{atomic-restriction}$\cross 2$]
	\Rename{\phipie}{
		\atev{\syncede}
	}
	\nondet 
	\Rename{\phipie}{(\omsid \Seq \top)}
	\Equals*[from \refeqn{defn-phipiE-in} and \refeqn{defn-phipiE-out}; \refeqn{omsidtop-elim}; simplify]
	\ateve
\]
\end{proof}

Recalling that CSP's
prefixing operator $e \fun c$ is defined as $\ateve \Seq c$, 
\reflemma{prefixed-interleaving} and \reflemma{atomic-sync} form the basis for proving communication axioms of CSP such
as the following.
\begin{eqnarray} 
	(e \fun c_1 \plEe e \fun c_2) &=& e \fun (c_1 \plEe c_2)
	\label{eqn-csp-synchronise}
	\\
	(e \fun c_1 \plEe f \fun c_2) &=& f \fun ((e \fun c_1) \plEe c_2)
	\label{eqn-csp-interleave}
\end{eqnarray}

\newcommand{\renamingHE}{{\cal H}_E}

The \emph{hiding} operator of CSP, $\Hide{E}{c}$, affects program steps, renaming events in $E$
to silent events.  
This operator may be encoded straightforwardly as a renaming.
\begin{equation}
	\Hide{E}{c}
	\sdef
	\Rename{\renamingHE}{c}
	\quad
	\mbox{where} 
	\quad
		\begin{array}[t]{rcll}
			\renamingHE(\spstepe) &=& \ssilent
			\quad
			& \mbox{if $e \in E$}
			\\
			\renamingHE(a) &=& a
			\quad
			&\mbox{for all other steps}
		\end{array}
\end{equation}
Hence, following from \refeqn{csp-synchronise},
\[
	\Hide{\{e\}}{(e \fun c_1 \plE{\{e\}} e \fun c_2)} = \silent \fun \Hide{\{e\}}{(c_1 \plE{\{e\}} c_2)}
\]

\subsection{Combining states and events}

Specification languages such as Circus \cite{SemanticsOfCircus}, PAT \cite{PATCSP} and CSP$_{\sigma}$ \cite{IFM09}
combine CSP and state.  
Below we define a simple unified language as an interpretation of our algebra,
which we base on combining elements from
the earlier definitions, arbitrarily choosing CSP-style synchronisation rather than CCS-style.

The instantiation of the Boolean algebra for steps is in this case is a set of triples,
$\power(\Sigma \cross Event \cross \Sigma)$.
Using this instantiation we may describe atomic steps in which
the event can depend on the initial state, and the final state can depend on both
the initial state and the event. 
Synchronisation on program steps may be straightforwardly defined as below,
where
$T_1$ and $T_2$ are sets of state-event-state triples.
\begin{eqnarray}
	\cpstep{T_1} \pl \cpstep{T_2}
	&=&
	\cpstep{\{(\sigma, \syncede, \sigma') | (\sigma, e, \sigma') \in T_1
	\int T_2\}}
	\label{eqn-interp-csps}
\end{eqnarray}
This definition allows identical steps to be conjoined if they synchronise on
their events, and otherwise steps must be interleaved.  Renamings and
restrictions may be defined in this unified language to affect only the
event part of the steps.

If we set 
the event part to $Event = \{\silent\}$ then the interpretation \refeqn{interp-csps} collapses to that in
\refsect{rely-guarantee}, 
while choosing the state space to be the unit type (with just one element)
collapses the definitions to the interpretation in \refsect{csp-encoding}.

\section{Related Work}\label{S-related-work}

Our Synchronous Refinement Algebra (SRA) 
compares to Concurrent Kleene Algebra (CKA)
\cite{DBLP:journals/jlp/HoareMSW11} in that both extend a sequential algebra to allow
for reasoning about parallel composition. 
Synchronous Kleene Algebra (SKA) \cite{Pris10} is also based on Kleene Algebra 
but, unlike CKA, it adds tests and a synchronous parallel operator based on that of
Milner's SCCS \cite{Milner83}.
Both CKA and SKA are based on Kleene algebra and 
hence only support finite iteration and partial correctness. 
In comparison, our SRA
supports general fixed points 
and hence recursion and both finite and infinite iteration.
The richer structure of DRA contains a sub-lattice of commands below $\Chaos$
(see Fig.\ \ref{lattices}) that includes 
assertions (and hence preconditions in the relational interpretation)
and
assumptions (and hence rely commands),
and allows the weak conjunction operator, $\together$, to be distinguished from strong conjunction, $\sqcup$.
All these constructs are needed to faithfully represent rely/guarantee theory.

CKA is also applied to rely/guarantee rules \cite{DBLP:journals/jlp/HoareMSW11}
but they define a Jones-style 5-tuple  (as in Section~\ref{S-Hoare-logic})
in terms of two separate refinement conditions,
whereas in our approach the existing (single) refinement relation can be used directly.
In Jones' theory, a guarantee has to be satisfied only from initial states satisfying the precondition of the program,
and further, if its rely condition is broken by the environment, the program can abort.
However, in the CKA framework, the guarantee has to always be maintained by the program, irrespective
of what the initial state is and how the environment is behaving;
that over restricts the set of possible implementations.
Our theory faithfully reflects Jones' approach.

Our algebra of atomic steps makes use of a synchronous parallel operator similar to that 
in SCCS \cite{CaC} and 
in SKA \cite{Pris10}
but 
it differs in two main ways: 
\begin{itemize}
\item
instead of atomic actions being separate from commands (as in SCCS and SKA),
they are treated as a sub-algebra of commands within SRA
and
\item
while both 
SCCS and SKA
explicitly define composition of atomic steps (their $\times$ operator),
our parallel operator is used directly on atomic steps (because they are commands)
and its definition on atomic steps is left open to allow multiple interpretatiotns.
\end{itemize}

Prensa Nieto has encoded rely-guarantee theory in Isabelle/HOL \cite{PrensaNieto03}.
Her language disallowed nested parallelism but allowed a multi-way parallel at the top level only, 
while here nested parallelism is allowed 
and hence a multi-way parallel can be defined in terms of a binary parallel.
Her work made use of a state-based operational semantics and 
showed the soundness of Hoare-style rely-guarantee quintuple rules with respect to the semantics directly.
In comparison, our approach is axiomatically based and 
to show soundness with respect to a semantics, 
we need show the axioms hold in the semantic model \cite{DaSMfaWSLwC}.
Our approach follows the more general refinement calculus style, rather than quintuples,
and hence it is easier to develop new laws.
It is also more abstract than that of Prensa Nieto and hence more widely applicable.

\section{Conclusion}\label{S-conclusion}

The main aim of this research is to provide mechanised support for the
verification/derivation of concurrent programs.
The approach taken is to develop a set of algebraic theories,
where we reason about programming operators at an abstract level
and then instantiate the theories, perhaps multiple times, 
to build our overall theory.
Just as mathematics has benefited from utilising abstract algebras,
such as semi-groups, lattices and boolean algebras,
we also benefit from using these and more programming-specific algebras.
While our initial aim was to provide a theory to support relational rely-guarantee concurrency,
by focusing on the abstract properties of operators and
careful structuring of theories 
-- an iterative process based on feedback from developing the theories --
we have come up with a collection of abstract theories, 
most of which are applicable to other contexts.

The basis of our theory is a Demonic Refinement Algebra (DRA) similar to that of von Wright \cite{Wright04}.
That gives us a simple program algebra over a complete lattice of commands $\Commands$
with sequential composition.
Unlike von Wright, we directly use fixed point operators to define iteration operators,
rather than explicitly axiomatising the iteration operators themselves;
that allows us to make more general use of recursion.
To the DRA we add a boolean sub-algebra of tests, $\Tests$, 
in a manner similar to that of Kozen's Kleene Algebra with Tests (KAT) \cite{kozen97kleene},
and a boolean sub-algebra of atomic steps, $\AtomicSteps$,
following our approach in \cite{FM2016atomicSteps}.

An innovation in the current paper is to axiomatise an abstract synchronisation operator, $\sync$,
and prove a set of ``synchronisation'' laws in the abstract theory.
The abstract theory is then instantiated for
parallel composition ($\parallel$),
weak conjunction ($\together$)
and the lattice supremum operator ($\sqcup$),
thus immediately giving a set of synchronisation laws for each of these operators 
without requiring further proof.

The atomic step commands, $\AtomicSteps$, are initially treated without any internal structure
other than that imposed by the Boolean algebra and 
the existence of an identity for the synchronisation operators.
Unlike in \cite{FM2016atomicSteps} which went directly to a relational interpretation,
we then give $\AtomicSteps$ more structure by identifying a subset of program steps $\PCommands$ and
another subset of matching environment steps $\ECommands$.
All atomic steps can be constructed as a non-deterministic choice of program and environment steps.
That structure is sufficient for us to define abstract rely and guarantee commands,
as well as an abstract version of Morgan's specification command \cite{TSS}.
We can then derive abstract versions of rely-guarantee concurrency laws,
including an abstract version of the parallel introduction law (Law \ref{L-introduce-parallel}),
one of the core laws in the rely-guarantee approach to concurrency.
Law~\ref{L-introduce-parallel} is a generalisation of the parallel introduction law of Jones
in the sense that it applies for any terminating commands satisfying $c \parallel \Term = c$.
The version for a relation postcondition specification (\ref{L-introduce-parallel-rspec})
corresponds to that of Jones
but other instances are possible, 
for example, a specification command using more expressive constraints on the behaviour 
during execution, such as via the use of possible values notation \cite{PVEaCfC}.
Exploration of such alternative specifications is a goal of future work.

Because our theory follows the refinement calculus approach of treating
precondition assertions, specifications, relies and guarantees as commands,
in comparison with using Hoare-style quintuples,
it is easier to develop a range of useful lemmas for each construct in isolation as well as 
for combinations of them.
For example, Lemma~\ref{L-introduce-parallel} for introducing parallelism builds on the
simpler Lemmas \ref{L-weaken-rely}, \ref{L-rely-guar} and \ref{L-parallel-guarantee}
involving properties of relies and guarantees.
It is also simpler to develop new laws, for example, 
it is straightforward to develop a version of the parallel introduction lemma
that includes preconditions.

Only at the final stage do we instantiate our test and atomic steps theories.
Firstly, for shared-memory concurrency,
where tests form a boolean algebra over sets of states,
and program and environment steps are boolean algebras over sets of pairs of states, i.e.\ relations.
Secondly, for process algebras 
tests are instantiated with the booleans, $\bool$,
and program and environment atomic steps are instantiated as boolean algebras
over sets of events.
All the laws derived in the abstract theory are applicable to both instantiations.
For the process algebra instantiation, by choosing appropriate sets of events
and defining how program steps combine via the parallel operator 
one can encode different process algebras such as CCS/SCCS and CSP.
One discovery was that the restriction operator in CSS corresponds to a 
guarantee (with a set complement because it excludes events).

The theory has been encoded in Isabelle/HOL \cite{IsabelleHOL}.
The Isabelle encoding is extensive, 
comprising of 26 theory files with 
over 350 proven lemmas,
and has evolved in unison with the development of our theory.
It makes extensive use of Isabelle/HOL locales to 
axiomatise the operators and structure the theories.

Overall, we think the algebraic approach has succeeded admirable in our quest
to mechanise rely-guarantee concurrency, 
but more than that it has surprised us with its ability to unify 
what are traditionally treated as separate approaches to concurrency.

\paragraph{Acknowledgements.}
This work has benefited from input from 
Cliff Jones, 
Kim Solin,
and 
Andrius Velykis.

\bibliographystyle{plain}
\bibliography{main}

\newpage
\appendix

\newenvironment{lemmaproof}[1]{\par\noindent\textbf{Lemma \ref{L-#1} (#1)}\itshape}{}
\newenvironment{theoremproof}[1]{\par\noindent\textbf{Theorem \ref{T-#1} (#1)}\itshape}{}
\def\arraystretch{1.5}

\newenvironment{new-proof}{\paragraph{{\bf Proof.}}}{\hfill$\Box$}

\section{Proofs for inspection}
\label{S-proofs}

For all lemmas we assume $a$ and $b$ to be atomic steps, $c$ and $d$ arbitrary commands,
and $t$ and $t'$ tests.
\vspace*{2ex}

\begin{lemmaproof}{atomic-iteration-finite}
\label{P-atomic-iteration-finite}
\begin{eqnarray*}
 \Fin{a} \Seq c \sync \Fin{b} \Seq d 
  & = & \Fin{(a \sync b)} \Seq ((c \sync \Fin{b} \Seq d) \nondet (\Fin{a} \Seq c \sync d))
\end{eqnarray*}
\end{lemmaproof}

\begin{proof}
The proof relies on (\ref{L-finite-iteration}), i.e., $\Fin{a} = \Nondet_{i \in \nat} a^i$. 
The notation $\Nondet_{i,j \in \nat}^{i \leqslant j} c_{i,j}$ stands for the choice of $c_{i,j}$
 over all natural numbers $i$ and $j$, such that $i \leqslant j$.
\begin{align*}
    \Fin{a} \Seq c \sync \Fin{b} \Seq d
    & = (\Nondet_{i \in \nat} a^i \Seq c) \sync (\Nondet_{j \in \nat} b^j \Seq d) \\
    & = \Nondet_{i,j \in \nat} (a^i \Seq c \sync b^j \Seq d)\\[1ex]
    & = \Nondet_{i,j \in \nat}^{i \leqslant j} (a^i \Seq c \sync b^i \Seq b^{j-i} \Seq d)
         \nondet \Nondet_{i,j \in \nat}^{i \geqslant j} (a^j \Seq a^{i-j} \Seq c \sync b^{j} \Seq d)\\
   & = \Nondet_{i, k \in \nat}(a \sync b)^i \Seq (c \sync b^k \Seq d) 
         \nondet \Nondet_{j, k \in \nat}(a \sync b)^j \Seq (a^k \Seq c \sync d)\\
   & = (\Nondet_{i \in \nat} (a \sync b)^i)  \Seq \Nondet_{k \in \nat}(c \sync b^k \Seq d)
         \nondet (\Nondet_{j \in \nat} (a \sync b)^j) \Seq \Nondet_{k \in \nat}(a^k \Seq c \sync d) \\
   & =  \Fin{(a \sync b)} \Seq ((c \sync \Nondet_{k \in \nat}b^k \Seq d)
         \nondet (\Nondet_{k \in \nat}a^k \Seq c \sync d)\\
   & = \Fin{(a \sync b)} \Seq ((c \sync \Fin{b} \Seq d) \nondet (\Fin{a} \Seq c \sync d))
\end{align*}
\end{proof}

\begin{lemmaproof}{atomic-iteration-finite-infinite}
\label{P-atomic-iteration-finite-infinite}
\[
 ~~\Fin{a} \Seq c \sync \Inf{b} = \Fin{(a \sync b)} \Seq (c \sync \Inf{b})
\]
\end{lemmaproof}

\begin{proof}
Note that, by unfolding law (\ref{L-infinite-unfold-power}), $\Inf{b} = b^i \Seq \Inf{b}$ for any $i \in \nat$.
The proof also uses Lemma~\ref{L-atomic-iteration-power}.
\begin{displaymath}
 \begin{array}{rclclcl}
\Fin{a} \Seq c \sync \Inf{b} 
  & = & (\Nondet_{i \in \nat} a^i \Seq c) \sync \Inf{b}\\
  & = & \Nondet_{i \in \nat}(a^i \Seq c \sync \Inf{b}) \\
  & = & \Nondet_{i \in \nat}(a^i \Seq c \sync b^i \Seq \Inf{b}) \\
  & = &  \Nondet_{i \in \nat}(a \sync b)^i \Seq (c \sync \Inf{b})\\
  & = & \Fin{(a \sync b)} \Seq (c \sync \Inf{b})  
 \end{array}
\end{displaymath}
\end{proof}

\begin{lemmaproof}{atomic-iteration-either}
\label{P-atomic-iteration-either}
\begin{eqnarray*}
 \Om{a} \Seq c \sync \Om{b} \Seq d & = &
    \Om{(a \sync b)} \Seq ((c \sync \Om{b} \Seq d) \nondet (\Om{a} \Seq c \sync d))
\end{eqnarray*}
\end{lemmaproof}

\begin{proof} Note that, by (\ref{L-isolation}) and (\ref{L-infinite-annihilates}), 
$\Om{a} = \Fin{a} \nondet \Inf{a}$ and $\Inf{a} \Seq c = \Inf{a}$.
The proof uses also (\ref{L-iteration1}), 
and Lemmas \ref{L-atomic-iteration-finite} and \ref{L-atomic-iteration-finite-infinite}, 
and (\ref{A-atomic-infiter-sync}) for the synchronisation operator, i.e.\ $\Inf{a} \sync \Inf{b} = \Inf{(a \sync b)}$.
\begin{align*}
  \Om{a} \Seq c \sync \Om{b} \Seq d 
  &  =  (\Fin{a} \nondet \Inf{a}) \Seq c \sync (\Fin{b} \nondet \Inf{b}) \Seq d \\
  & =  (\Fin{a} \Seq c \sync \Fin{b} \Seq d) \nondet (\Fin{a} \Seq c \sync \Inf{b}) \nondet 
           (\Inf{a} \sync \Fin{b} \Seq d) \nondet (\Inf{a} \sync \Inf{b}) \\
  & =  \Fin{(a \sync b)} \Seq ((c \sync \Fin{b} \Seq d) \nondet (\Fin{a} \Seq c \sync d)) \nondet 
         \Fin{(a \sync b)} \Seq (c \sync \Inf{b}) \nondet \Fin{(a \sync b)} \Seq (\Inf{a} \sync d) \nondet \Inf{(a \sync b)} \\
  & =  \Fin{(a \sync b)} \Seq ((c \sync \Fin{b} \Seq d) \nondet (c \sync \Inf{b}) \nondet 
                                            ((\Fin{a} \Seq c \sync d) \nondet (\Inf{a} \sync d))) \nondet \Inf{(a \sync b)} \\
  & =  \Fin{(a \sync b)} \Seq ((c \sync (\Fin{b} \Seq d \nondet \Inf{b})) \nondet 
                                            ((\Fin{a} \Seq c  \nondet \Inf{a}) \sync d)) \nondet \Inf{(a \sync b)} \\
  & =  \Fin{(a \sync b)} \Seq ((c \sync \Om{b} \Seq d) \nondet
                                             (\Om{a} \Seq c  \sync d)) \nondet \Inf{(a \sync b)} \\
  & =  \Om{(a \sync b)} \Seq ((c \sync \Om{b} \Seq d) \nondet (\Om{a} \Seq c  \sync d))
\end{align*}
\end{proof}

\begin{lemmaproof}{iterations-with-abort}\label{P-iterations-with-abort}
Provided $c \sync \bot = \bot$, for any command $c$,
\[
  \Om{(a_0 \nondet a_1 \Seq \bot)} \sync \Om{b} = \Om{(a_0 \sync b)} \Seq (\Nil \nondet (a_1 \sync b) \Seq \bot)~.
\]
\end{lemmaproof}

\begin{proof}
\[
  \Om{(a_0 \nondet a_1 \Seq \bot)} \sync \Om{b}
 \Equals*[omega decomposition]
  \Om{(\Fin{a_0} \Seq a_1 \Seq \bot)} \Seq \Om{a_0} \sync \Om{b}
 \Equals*[omega unfolding]
  (\Om{a_0} \nondet \Fin{a_0} \Seq a_1 \Seq \bot) \sync \Om{b}
 \Equals*[distribute]
  (\Om{a_0} \sync \Om{b}) \nondet (\Fin{a_0} \Seq a_1 \Seq \bot \sync \Om{b})
 \Equals*[by (\ref{omega-sync-omega}) and isolation, i.e.\ $\Om{b} = \Fin{b} \nondet \Inf{b}$] 
  \Om{(a_0 \sync b)} \nondet (\Fin{a_0} \Seq a_1 \Seq \bot \sync \Fin{b}) \nondet (\Fin{a_0} \Seq a_1 \Seq \bot \sync \Inf{b})
 \Equals*[by Lemmas \ref{L-atomic-iteration-finite} and \ref{L-atomic-iteration-finite-infinite}]
  \Om{(a_0 \sync b)} \nondet 
    \Fin{(a_0 \sync b)} \Seq ((a_1 \Seq \bot \sync \Fin{b}) \nondet (\Fin{a_0} \Seq a_1 \Seq \bot \sync \Nil)) \nondet
    \Fin{(a_0 \sync b)} \Seq (a_1 \Seq \bot \sync \Inf{b})
 \Equals*[simplifying assuming $c \sync \bot = \bot$]
  \Om{(a_0 \sync b)} \nondet 
    \Fin{(a_0 \sync b)} \Seq (a_1 \sync b) \Seq \bot \nondet 
    \Fin{(a_0 \sync b)} \Seq (a_1 \sync b) \Seq \bot 
 \Equals*[non-deterministic choice is idempotent; isolation]
  \Om{(a_0 \sync b)} \Seq (\Nil \nondet (a_1 \sync b) \Seq \bot)
\]
\end{proof}

From Corollary~\ref{C-atomic-iteration-either} we can derive the following two corollaries
which will be used in later proofs.

\begin{corollary}[atomic-iter-prefix-sync-nil]\label{L-atomic-iter-prefix-sync-nil}
For any atomic steps $a, b, b_1 \in \AtomicSteps$, and commands $c \in \Commands$,
\[
\Om{a}\sync \Om{b}\Seq b_1\Seq c 
                          = \Om{(a \sync b)} \Seq (a \sync b_1) \Seq (\Om{a} \sync c)
\]
\end{corollary}

\begin{corollary}[atomic\_iter\_prefix\_sync\_atomic]\label{L-atomic-iter-prefix-sync-atomic}
For any atomic steps $a, a_1, b, b_1 \in \AtomicSteps$, and commands $c, d \in \Commands$,
\[
\Om{a} \Seq a_1 \Seq c \sync \Om{b} \Seq b_1\Seq d 
                          = \Om{(a \sync b)} \Seq ((a_1 \sync b_1) \Seq (c \sync d) 
                             \nondet (a_1 \sync b)\Seq (c \sync \Om{b}\Seq b_1 \Seq d)
                             \nondet (a \sync b_1)\Seq (\Om{a}\Seq a_1\Seq c \sync d))
\]
\end{corollary}

With Corollary~\ref{L-atomic-iter-prefix-sync-nil} and (\ref{A-together-atomic}) we have
\begin{lemma}[assump-help1]\label{L-assump-help1}
For any atomic steps $a, b \in \AtomicSteps$,
\[
\Om{a}\together\Om{b}\Seq\anegate b\Seq\bot 
                           = \Om{(a \join b)}\Seq(a \join \anegate b) \Seq\bot
\]
\end{lemma}

With Corollary~\ref{L-atomic-iter-prefix-sync-atomic}  we have

\begin{lemma}[assump-help2]\label{L-assump-help2}
For any atomic steps $a, b \in \AtomicSteps$,
\[
\Om{a}\Seq\anegate a \Seq\bot \together \Om{b}\Seq\anegate b\Seq\bot
                 = \Om{(a \join b)}\Seq ((\anegate a \join \anegate b) \join (\anegate a \join b) 
                   \nondet (a \join \anegate b))\Seq\bot
\]
\end{lemma}

Since atomic commands form a Boolean algebra and with (\ref{A-together-atomic}) we can derive the following
two lemmas.
\begin{lemma}[helper2]\label{L-helper2}
For any atomic steps $a, b \in \AtomicSteps$,
\[
 (\anegate a \together b) \nondet (a \together \anegate b) \nondet (\anegate a \together \anegate b)
   =  \anegate(a \together b)
\]
\end{lemma}

\begin{lemma}[assump-help3]\label{L-assump-help3}
For any atomic steps $a, b \in \AtomicSteps$,
\[
(a \join \anegate b) \nondet (\anegate a \join b) \nondet (\anegate a \join \anegate b) 
                       = \anegate(a \join b)
\]
\end{lemma}

\pagebreak[3]

Lemmas~\ref{L-assump-help1}, \ref{L-assump-help2}, and \ref{L-assump-help3} are used in the
proof of the following lemma, which in turn is valuable when proving the conjunction of relies 
(Lemma~\ref{L-combine-relies}).

\begin{lemmaproof}{assume-iter-conj-assume-iter}
\label{P-assume-iter-conj-assume-iter}
For any atomic steps $a, b \in \AtomicSteps$,
\[
\Om{(\Assume{a})} \together \Om{(\Assume{b})} ~=~ \Om{(\Assume{(a \join b)})}
\]
\end{lemmaproof}

\begin{proof}
\[
\Om{(\Assume{a})} \together \Om{(\Assume{b})}
 \Equals*[with Lemma~\ref{L-iterated-assumption} and (\ref{A-seq-distr-left})]
    (\Om{a} \nondet \Om{a}\Seq \anegate{a} \Seq \bot) 
      \together (\Om{b} \nondet \Om{b}\Seq \anegate{b} \Seq \bot)
 \Equals*[with (\ref{A-seq-distr-right}) and (\ref{A-sync-comm})] 
    (\Om{a} \together \Om{b}) 
       \nondet (\Om{a} \together \Om{b}\anegate{b}\Seq\bot) 
       \nondet (\Om{a}\anegate{a}\Seq\bot \together \Om{b}) 
       \nondet (\Om{a}\anegate{a}\Seq\bot \together \Om{b}\anegate{b}\Seq\bot)
 \Equals*[(\ref{A-together-atomic}), Corollary~\ref{L-atomic-either-sync}, and
           Lemma~\ref{L-assump-help1}]
    \Om{(a \join b)}
        \nondet \Om{(a \join b)}\Seq(a \join \anegate{b})\Seq \bot
        \nondet \Om{(a \join b)}\Seq (\anegate{a} \join b)\Seq\bot
        \nondet (\Om{a}\anegate{a}\Seq \bot \together \Om{b}\anegate{b}\Seq\bot)
 \Equals*[with Lemma~\ref{L-assump-help2}]
     \Om{(a \join b)}
       \nondet \Om{(a \join b)}\Seq(a \join \anegate{b})\Seq \bot
        \nondet \Om{(a \join b)}\Seq (\anegate{a} \join b)\Seq\bot
        \nondet \Om{(a \join b)}\Seq
     ((\anegate{a} \join \anegate{b}) \nondet (a \join \anegate{b}) \nondet (\anegate{a} \join b))\Seq\bot\\
 \Equals*[with (\ref{A-seq-distr-right})]
     \Om{(a \join b)}
        \nondet \Om{(a \join b)}\Seq(a \join \anegate{b})\Seq \bot
        \nondet \Om{(a \join b)}\Seq (\anegate{a} \join b)\Seq\bot
        \nondet \Om{(a \join b)}\Seq (\anegate{a} \nondet \anegate{b})\Seq\bot
 \Equals*[with (\ref{A-seq-distr-left}) and (\ref{A-seq-distr-right})]
   \Om{(a \join b)}\Seq(\Nil \nondet ((a \join \anegate{b}) \nondet (\anegate{a} \join b) 
                             \nondet \anegate{(a \join b)})\Seq \bot)
  \Equals*[with Lemma~\ref{L-assump-help3}]
   \Om{(a \join b)}\Seq(\Nil \nondet \anegate(a \join b)\Seq\bot)
 \Equals*[with Lemma \ref{L-iterated-assumption}]
   \Om{(\Assume{(a \join b)})}
\]
\end{proof}

\def\arraystretch{1}

\end{document}